\pdfoutput=1

\documentclass[a4paper,11pt,onecolumn]{article}
\usepackage[utf8]{inputenc}
\usepackage[margin=1in]{geometry}
\usepackage[english]{babel}
\usepackage[perpage,symbol*]{footmisc}
\usepackage{comment}
\usepackage{amsmath,amssymb,amsthm}
\usepackage{todonotes}
\usepackage{physics}
\usepackage{tabularx}
    \newcolumntype{Y}{>{\centering\arraybackslash}X}
    \newcolumntype{Z}[1]{>{\centering\arraybackslash}p{#1\textwidth}}
\usepackage{mathtools}
\usepackage{multicol}

\usepackage{booktabs}
\usepackage{color, colortbl}
\usepackage[ruled]{algorithm2e}
\usepackage{xparse}

\usepackage{footnote}
\usepackage[section]{placeins}
\PassOptionsToPackage{hyphens}{url}
\usepackage{hyperref}
\usepackage[capitalize, noabbrev]{cleveref}
\usepackage{chemformula}
\usepackage{paralist}

\usepackage{pgfplots}
\usepgfplotslibrary{fillbetween}
\usepgfplotslibrary{external}
\usetikzlibrary{backgrounds,calc,arrows.meta,positioning, automata}
\usepackage{tikzscale}
\pgfplotsset{compat=1.14}
\usepackage{subcaption}
\usepackage{thm-restate}
\usepackage{url}
\makeatletter
\g@addto@macro{\UrlBreaks}{\UrlOrds}
\makeatother
\makeatletter
\def\url@leostyle{%
  \@ifundefined{selectfont}{\def\UrlFont{\sf}}{\def\UrlFont{\small\ttfamily}}}
\makeatother
\usepackage{aliascnt}

\RequirePackage[%
  babel,%
  final,%
  expansion=alltext,%
  protrusion=alltext-nott]{microtype}%

\graphicspath{{./img/}}

\theoremstyle{plain}
\newtheorem{theorem}{Theorem}
\newaliascnt{lemma}{theorem}
\newtheorem{lemma}[theorem]{Lemma}
\newtheorem{corollary}[theorem]{Corollary}

\newtheorem{definition}[theorem]{Definition}
\theoremstyle{definition}

\theoremstyle{remark}

\newtheorem*{note*}{Note}

\newtheorem*{remark*}{Remark}

\title{A Strategic Routing Framework and\\
  Algorithms for Computing Alternative Paths}

\setfnsymbol{lamport}
\author{%
    Thomas Bläsius\footnote{Hasso Plattner Institute, University of Potsdam,
    Germany, \texttt{firstname.lastname@hpi.de}}
    \and
    Maximilian Böther\footnote{Hasso Plattner Institute, University of Potsdam,
    Germany, \texttt{firstname.lastname@student.hpi.de}}
    \and
    Philipp Fischbeck\footnotemark[1]
    \and
    Tobias Friedrich\footnotemark[1]
    \and
    Alina Gries\footnotemark[2]
    \and
    Falk Hüffner\footnote{TomTom Location Technology Germany GmbH, \texttt{falk.hueffner@tomtom.com}}
    \and
    Otto Kißig\footnotemark[2]
    \and
    Pascal Lenzner\footnotemark[1]
    \and
    Louise Molitor\footnotemark[1]
    \and
    Leon Schiller\footnotemark[2]
    \and
    Armin Wells\footnotemark[2]
    \and
    Simon Wietheger\footnotemark[2]
}
\date{}

\newcommand{\set}[1]{\{#1\}}
\newcommand{\Tcost}{\mathcal{C}}

\definecolor{lightgreen}{rgb}{0.75,0.92,0.61}

\newcommand{\labelname}[1]{
  \def\@currentlabelname{#1}}

\newif\ifarxivVersion
\arxivVersiontrue

\newlength{\punctuationfootlength}
\newcommand{\punctuationfootnote}[2]{#2\settowidth{\punctuationfootlength}%
{#2}\hspace{-0.5\punctuationfootlength}\footnote{#1}}

\definecolor[named]{Greenery}{rgb}{0.533, 0.69, 0.294}

\definecolor{brightlavender}{rgb}{0.75, 0.58, 0.89}

\begin{document}
\maketitle
\begin{abstract}
\noindent Traditional navigation services find the fastest route for a single driver.
Though always using the fastest route seems desirable for every individual,
selfish behavior can have undesirable effects such as higher energy consumption
and avoidable congestion, even leading to higher overall and individual travel
times. In contrast, strategic routing aims at optimizing the traffic for all
agents regarding a global optimization goal.  We introduce a framework to
formalize real-world strategic routing scenarios as algorithmic problems and
study one of them, which we call \emph{Single Alternative Path (SAP)}, in
detail.  There, we are given an original route between a single
origin--destination pair. The goal is to suggest an alternative route to all
agents that optimizes the overall travel time under the assumption that the
agents distribute among both routes according to a psychological model, for
which we introduce the concept of Pareto-conformity.  We show that the SAP
problem is NP-complete, even for such models. Nonetheless, assuming
Pareto-conformity, we give multiple algorithms for different variants of SAP,
using multi-criteria shortest path algorithms as subroutines.  Moreover, we
prove that several natural models are in fact Pareto-conform.  The
implementation of our algorithms serves as a proof of concept, showing that SAP
can be solved in reasonable time even though the algorithms have exponential
running time in the worst case.


\end{abstract}

\section{Introduction}\label{sec:introduction}

Commuting is part of our daily lives. Street congestion, traffic jams
and pollution became an increasingly large issue in the last few
decades. In German cities, these effects caused costs of about 3
billion euros in 2019~\cite{Inrix2020}.
Many traffic jams in cities could have been avoided by better route
choice. Partly this is because of non-optimal route choices by
individuals due to bounded rationality and route preferences other
than ``fastest''~\cite{zhu_levinson_2015}. However, even with
individually optimal route choice, average travel time can be
substantially worse compared to a system optimum where all routes are
centrally assigned~\cite{roughgarden_2005}.
%
Thus, there is an opportunity for improving traffic via
\emph{strategic routing} where (re)routing recommendations are created
by traffic authorities and taken into account by the driver's routing
system. More precisely, we speak of strategic routing when two
conditions are met:
\begin{enumerate}[(i)]
\item One or more routes are calculated to be proposed to \emph{more than
    one agent}, and
\item the quality of a set of proposed routes is being defined by a \emph{shared
    scoring} rather than scoring each agent individually.
\end{enumerate}

\noindent Recent research indicates that many drivers would accept
individually slower routes if this contributes to an overall reduction
in traffic~\cite{vanEssen2018, kroller_driver_2020}; additionally,
incentives such as free parking could be granted to those accepting
these routes, and future autonomous vehicles may be more amenable to
centralized control. Thus, (re)routing recommendations can have a strong
impact since they might be followed by a significant fraction of all drivers.

In the ongoing pilot research project
\emph{Socrates~2.0}, strategic routing is employed in the area of
Amsterdam~\cite{Socrates2Project}. For this, experts predefine
alternative routes and traffic conditions that trigger their
recommendation. This requires extensive work and monitoring, and does not capture well unusual traffic situations where
there might be several incidents at once causing delays. Thus, it is
desirable to automate this by formalizing strategic routing and
finding algorithms that calculate strategic routes.

\subparagraph*{Our Contribution.}

Strategic routing as defined above is not a single algorithmic problem
but rather a concept capturing numerous scenarios leading to different
problems.  In \Cref{sec:dimensions}, we provide a framework to guide
the process of formalizing real-world strategic routing scenarios.  We
apply it to one specific scenario, namely \emph{Single Alternative
  Path (SAP)}.  This scenario is inspired by the Amsterdam use case
mentioned above where congestion can be prevented by suggesting one
alternative route to all agents, e.g., via a variable-message sign.  We consider
different psychological models to determine how many agents follow the
suggestion.  Moreover, we consider variants of the SAP problem that
require the alternative to be more or less disjoint from the original
route. See~\Cref{sec:probl-form} for a formal definition.

To tackle SAP algorithmically, we introduce the concept of Pareto-conformity of
psychological models and, based on this, give various algorithms
in~\Cref{sec:sap-algos}.  As they use multi-criteria shortest path algorithms as
subroutine, they have an exponential worst-case running time but turn out to be
sufficiently efficient in practice; see our evaluation in~\Cref{sec:evaluation}.
Moreover, in this generality, we cannot hope for better worst-case bounds as SAP
is NP-hard, even for Pareto-conform psychological models; see \Cref{sec:np}.  In
\Cref{sec:psychmods}, we prove the Pareto-conformity of three natural
psychological models.  Our proofs actually hold for the more general and
abstract \emph{Quotient Model} that captures various additional models.  We
evaluate our algorithms in~\Cref{sec:evaluation}.  It serves as a proof of
concept that our algorithms have reasonable practical run times and yield
promising travel time improvements for instances in the traffic network of
Berlin.

\subparagraph*{Related Work.}\label{subsec:related-work}
There has been no unique understanding of strategic routing in research until
this point. Van Essen~\cite{vanEssen2018} uses a
choice-theoretical approach and concludes that individual route choice and travel information that stimulates non-selfish user behavior have a large impact on the
network efficiency. Kröller et
al.~\cite{kroller_driver_2020} investigate due to what kind of
incentives agents would deviate from the shortest-path route. Their
results show that certain incentives can increase the drivers' willingness of
taking detours. Moreover, they show that there is a high interest in services
providing alternative routes, and strategic routing is considered to have
the potential of solving traffic issues such as congestion and pollution.

For standard algorithmic techniques in efficient route planning, we refer to the
survey of Bast et al.~\cite{bast_route_2016}.
Köhler et al.~\cite{Koehler_Moehring_Skutella_2009} deal with finding static and also time-dependent
traffic flows minimizing the overall travel time. Also, as stated by Strasser~\cite{Strasser17}, routing with predicted congestion is well-studied, e.g., by
Delling and Wagner~\cite{Delling2009}, Demiryurek et al.~\cite{Demiryurek2010}, Delling~\cite{Delling2009_2} and Nannicini et al.~\cite{Nannicini2011}. Route planning with
alternative routes was investigated by Abraham et al.~\cite{Abraham2013} and Paraskevopoulos and
Zaroliagis~\cite{Paraskevopoulos2013}. They propose
algorithms that find alternative routes by evaluating properties with regard to
an original route.

Lastly, we emphasize that strategic routing is very different from selfish routing as proposed by Roughgarden and Tardos~\cite{Roughgarden_Tardos_2002}. In contrast to our global optimization approach, in selfish routing individual strategic agents select their routes to optimize their own travel times, given the route choices of other agents. While often static flows are considered in selfish routing, Sering and Skutella~\cite{Sering18} analyzed selfish driver behavior for a dynamic flow-over-time model. Another related selfish routing variant is Stackelberg routing~\cite{Stackelberg_Korilis, Stackelberg_Bonifaci, Stackelberg_Karakostas, Stackelberg_Bhaskar}, where an altruistic central authority controls a fraction of the traffic and first routes it in a way to improve the travel times for all other selfish agents which choose their route afterwards.

\section{A Framework for Strategic Routing}
\label{sec:dimensions}

In the following, we provide a framework that supports the
formalization of a given strategic routing scenario.  We employ a
two-step process.  The first step categorizes the scenario by
distilling its crucial aspects.  The second step transforms it into an
algorithmic problem.

\subsection{Categorization}
\label{sec:categorization}

Categorizing a scenario at hand boils down to answering the
following questions.
%

\subparagraph{What is the goal we aim to achieve?}  There are
different objectives one can pursue when routing strategically. A city
might be interested in reducing particulate matter emission in a
certain region. As a routing service provider, the goal could be to
minimize the travel time for as many customers as possible. A system
of centrally controlled autonomous vehicles might want to achieve a
minimum overall travel time.
%

\subparagraph{How can we influence the agents?}  How we recommend
routes determines which agents we can influence and whether we can
make different suggestions to different agents.
A city administration can put up signs to
influence all vehicles in a certain area, making the same suggestion
to each agent.  Navigation providers, on the other hand, can influence
only a limited number of vehicles but could make different suggestions
to different agents.
%

\subparagraph{How much control do we have over the agents?} The
willingness of users to follow an alternative route depends on the use
case. While a navigation provider cannot force its users to use a
specific route, and the acceptance of detouring depends heavily on the
additional length, there are scenarios where the suggested route will
always be accepted or agents end up in an equilibrium or in a
system-optimal distribution on the suggested routes.
%

\subparagraph{What is the starting situation?}  We either assume that
there is already existing traffic, or that we design traffic from
scratch. Although the former is certainly more common, the latter
applies to, e.g., the scenario of centrally controlled autonomous
vehicles.
%

\subparagraph{How do the uninfluenced agents react?} If only a
fraction of the traffic is routed strategically, the remaining traffic
might react with respect to the change. For instance, it is a valid
assumption that after some time, all traffic settles in an
equilibrium. Another simple assumption is that the other traffic does
not change at all.

\subsection{Problem Formalization}
\label{sec:probl-form}



In this section, we first propose a generic formalization whose degrees of freedom can then be filled to reflect a specific scenario. We focus on the \textit{Single Alternative Path} (SAP) scenario, which we study algorithmically in~\Cref{sec:sap-algos}. We use it as an example how fixing answers to the questions raised in~\Cref{sec:categorization} naturally fills the degrees of freedom.

\subparagraph{Generic Strategic Routing Considerations.}\label{subsec:problem-statement}
Let \(G = (V, E)\) be a directed graph. For every
pair of nodes \((s, t) \in V^2\), the \emph{demand} \(d: V^2 \rightarrow
\mathbb{Q}\) denotes the amount of traffic flow that has to be routed from \(s\)
to \(t\). For every edge \(e \in E\), let \(\tau_e \colon \mathbb{Q}_{\ge 0}
\rightarrow \mathbb{Q}_{> 0}\) be a monotonically increasing cost function.
For \(x \in \mathbb{Q}_{\ge 0}\), \(\tau_e(x)\) describes the costs for a single agent
traversing an edge \(e \in E\) while there is a traffic flow of~\(x\) vehicles
per unit of time on \(e\).

The solution to a strategic routing problem is a traffic distribution to paths in the network that routes agents according to \(d\).
Let
\(\mathcal{P}\) be the set of all simple paths in \(G\). 
By \(f \colon \mathcal{P} \rightarrow
\mathbb{Q}_{\geq 0}\) we denote the flow, where \(f(P)\) states the amount of traffic flow using path \(P\).
Extending the notion, let \(f(e) = \sum_{e \in P} f(P)\) be the total traffic flow on an edge \(e\).
For all \(x\in\mathbb{Q}_{\ge 0}\), let
\(\tau_P(x) = \sum_{e\in P} \tau_e(x)\) be the costs per agent on
\(P\) assuming that the total traffic on $P$ is $x$.

Paths are denoted as tuples of vertices, i.e., $(v_1, \dots, v_k)$
with $v_i \in V$ is a path if for
$1 < i \le k$, $(v_{i-1}, v_{i}) \in E$ .  In addition, we consider paths as edge sets and
use set operators, which also translates to the notion of cost
functions, e.g., for paths $P$ and $Q$ let
\(\tau_{P\cap Q}(x) = \sum_{e \in P\cap Q} \tau_e(x)\). 

\subparagraph{Means of Influence.}  In the SAP problem, we assume that
we can influence all agents on a given \emph{original} $st$-path $Q$
and suggest a single \emph{alternative} $st$-path $P$.  

\subparagraph{Starting Situation and Uninfluenced Traffic.}  We assume
that there is existing traffic that satisfies all demands and that uninfluenced agents stick with their previous routes.
Note that this allows us to integrate the uninfluenced traffic into the cost functions. Thus, we can formalize it as if there was no initial traffic and that all demands are equal to~$0$ except for the traffic on the
original route $Q$ which satisfies the demand $d(s, t) > 0$.  For
brevity, we denote $d = d(s, t)$.






\subparagraph{Level of Control.}\label{subsec:control_rate}

We assume that agents make their own decisions.
Given an original route \(Q\) and alternative \(P\) a
\emph{psychological model} determines the amount of flow $x_P$ on $P$.
The flow on $Q$ is then $d - x_P$.  We consider the following three
psychological models; see \Cref{sec:psychmods} for formal definitions.
The \emph{System Optimum} assumes agents distribute optimally with
respect to the optimization criterion defined below.  In the
\emph{User Equilibrium}~\cite{Wardrop1952} agents act selfishly
leading to an equilibrium where no agent can improve by unilaterally
changing their route~\cite{Roughgarden_Tardos_2002}.  In the
\emph{Linear Model} we assume that the willingness to choose \(P\) is
linearly dependent on the ratio of the costs on $Q$ and \(P\).

\subparagraph{Optimization Criterion.}  The optimization criterion
formalizes the goal to be achieved, which is the \emph{overall travel
  time} for SAP. Hence, we interpret the cost functions \(\tau_e\) as
latency functions, i.e., the time a single agent needs to traverse the
edge \(e\). 
%
%
In the SAP problem, we only consider one alternative \(P\) to an
original route \(Q\).  Assume that we have a flow of $x \in [0, d]$ on
$P$.  Then, the edges in $P \setminus Q$ have flow $x$, the edges of
$Q \setminus P$ have flow $d - x$ and the edges of $P \cap Q$ have
flow $d$.  Thus, the overall cost is
%
%
\begin{equation}\label{eq:cost-P}
  \Tcost_P(x) = x \cdot \tau_{P
    \setminus Q}(x) + (d - x) \cdot \tau_{Q\setminus P}(d-x) + d\cdot
  \tau_{P \cap Q}(d).
\end{equation}
For the value \(x_P\) determined by the psychological model, the
actual cost of an alternative route $P$ is $\Tcost_P(x_P)$, which we
abbreviate with $\Tcost_P$.  Let $\mathcal P$ be a set of alternative
paths.  Computing the path $P$ in $\mathcal P$ with optimal $\Tcost_P$
is called \emph{scoring $\mathcal P$}.

\subparagraph*{Summary and Problem Variants.}
To sum up the SAP problem, given a route \(Q\) from \(s\) to~\(t\), a demand $d$ of agents per unit of time and a psychological model, the SAP problem asks for the optimal alternative route \(P\) such that the overall travel time $\Tcost_P$ is minimized. 

In general~$P$ can have arbitrarily many overlaps with $Q$.
Additionally, we consider two variants of SAP, where we require the
routes to be more or less disjoint.  \emph{Disjoint Single Alternative
  Path (D-SAP)} requires $P$ and $Q$ to be completely disjoint.
Moreover, \emph{1-Disjoint Single Alternative Path (1D-SAP)} requires
$P \setminus Q$ to be a single connected path, i.e., $P$ diverts from
$Q$ at most once but can share the edges at the start and the end with
$Q$.

\section{Algorithms for Single Alternative Path}\label{sec:sap-algos}

Consider two alternative paths $P_1$
and $P_2$ with cost functions $\tau_{P_1}$ and $\tau_{P_2}$,
respectively.  Assume that for any amount of traffic $x \in [0, d]$, the cost of
$P_1$ is not larger than of $P_2$, i.e.,
$\tau_{P_1}(x) \le \tau_{P_2}(x)$.  It
seems intuitive that it is never worse to choose $P_1$ over $P_2$.
However, this is not quite right for two reasons. First, it
does not hold for arbitrary
psychological models, which determine the amount of agents (potentially in a somewhat degenerate fashion) who choose $P_1$ and $P_2$, respectively, instead of the original route $Q$. Secondly, if the alternative route $P_1$ shares
many edges with the original route $Q$ it has only little potential to
distribute traffic, whereas the seemingly worse alternative $P_2$
could do better in this regard.

We resolve the first issue by
defining a property that we call (weak) Pareto-conformity.
Moreover, in~\Cref{sec:psychmods}, we show for various
psychological models that they are in fact Pareto-conform.  To resolve
the second issue with shared edges, we introduce a notion of dominance
between paths that takes the overlap with $Q$ into account.

Let $\tau_1$ and $\tau_2$ be two cost functions defined on the interval $[0, d]$
and let $\tau_i'$ denote the derivative of $\tau_i$\punctuationfootnote{In the
    remainder, we implicitly assume all cost functions to be only defined on
    $[0, d]$, e.g., $\tau_1 \le \tau_2$ means $\tau_1(x) \le \tau_2(x)$ for all
$x \in [0, d]$.  Also, we implicitly assume functions to be differentiable.}.
For two alternative paths $P_1$ and $P_2$, we say that $P_1$ \emph{dominates}
$P_2$, denoted by $P_1 \preceq P_2$, if $\tau_{P_1} \le \tau_{P_2}$ and
$\tau_{P_1 \cap Q}' \le \tau_{P_2 \cap Q}'$.  Note that, if $P_1 \cap Q = P_2
\cap Q$, then this simplifies to $\tau_{P_1} \le \tau_{P_2}$.

Intuitively, the requirement $\tau_{P_1} \le \tau_{P_2}$ indicates that $P_1$ is
the cheaper alternative compared to $P_2$.  Moreover, concerning $\tau_{P_1 \cap
Q}' \le \tau_{P_2 \cap Q}'$, note that this is equivalent to $\tau_{Q \setminus
P_1}' \ge \tau_{Q \setminus P_2}'$.  Thus, bypassing $Q \setminus P_1$ by using
$P_1$ saves more than if the same amount of drivers bypasses $Q \setminus P_2$
by using $P_2$.  With this, we can define the above-mentioned Pareto-conformity.

\begin{definition}\label{def:pareto-conformity}
  A psychological model is \emph{Pareto-conform} if $P_1 \preceq P_2$
  implies $\Tcost_{P_1} \le \Tcost_{P_2}$. It is \emph{weakly
    Pareto-conform} if this holds for paths that have equal
  intersection with $Q$.
\end{definition}

To simplify notation, we assume without loss of generality that there
are no two different paths $P_1$ and $P_2$ with $P_1 \preceq P_2$ and
$P_2 \preceq P_1$.  This can, e.g., be achieved by slight perturbation
of the cost functions, or by resolving every tie arbitrarily.


In the following we give different algorithms for the SAP, 1D-SAP and D-SAP problems.  The
algorithms involve solving one or more multi-criteria shortest path
problems as subroutine.
Algorithms for this problem range from the fundamental examination of the bicriteria case~\cite{Hansen_1980} to the usage of speed-up techniques~\cite{Delling_Wagner_2009, Mandow_DeLaCruz_2008} in the multi-criteria case.
One such algorithm is the multi-criteria
Dijkstra, which has exponential run
time in the worst case~\cite{Martins_1984} but is known to be
efficient in many practical applications~\cite{ParetoFeasible}.

The algorithms we present first
(Sections~\ref{sec:single-altern}--\ref{sec:disjoint-single-alt-path}) require
solving only a single multi-criteria shortest path problem, with the
D-SAP setting requiring fewer criteria than SAP and 1D-SAP. In
Sections~\ref{sec:1d-sap-dp} and~\ref{sec:dynamic-program-sap}, we
propose approaches that require multiple such searches.
Though the former seems preferable, the latter has some advantages.
It requires fewer criteria in the multi-criteria sub-problems, it
requires only weak Pareto-conformity for the 1-disjoint setting, and
it allows for easy parallelization.  Our experiments in~\Cref{sec:evaluation} indicate that the variants requiring fewer criteria
are often faster for long routes.

\subsection{Reduction to Multi-Criteria Shortest Path}\label{sec:single-altern}

We are now ready to solve SAP.~\Cref{def:pareto-conformity} directly yields the following
lemma.

\begin{lemma}\label{lem:sap-optimal-is-not-dominated}
For any instance of SAP with a Pareto-conform psychological model,
	there exists an optimal solution that is not dominated by any other alternative.
\end{lemma}

Thus, to solve SAP, it suffices to find all alternative paths that
are not dominated by other paths, and then choose the best among
these potential solutions.  We reduce the problem of computing the set
of potential solutions to a multi-criteria shortest path problem.  In
such a problem, each path corresponds to a point $p \in \mathbb Q^k$,
where the entry at the $i$-th position of $p$ is the cost of the path
with respect to the $i$-th criterion.  One then searches for all
solutions that are not Pareto dominated by other solutions.  For two
points $p_1, p_2 \in \mathbb Q^k$, $p_1$ \emph{Pareto dominates} $p_2$
if $p_1 \le p_2$ component-wise.  Finding all solutions that are not
Pareto dominated is the previously mentioned multi-criteria shortest
path problem.  How the transformation to a multi-criteria problem
exactly works depends on the cost functions.

Assume for now that $\tau(x) = a x^2 + b$ for positive $a$ and $b$.
We call the family of cost functions of this form \emph{canonical cost
  functions}.  It is closed under addition.  Thus, the cost function
of each path is also a canonical cost function.  Note that two
different canonical cost functions intersect in at most one point on
$[0, d]$.  Thus, we have $\tau_1 \le \tau_2$ if and only if
$\tau_1(0) \le \tau_2(0)$ and $\tau_1(d) \le \tau_2(d)$.  It follows
that requiring $\tau_1 \le \tau_2$ is equivalent to saying that
$(\tau_1(0), \tau_1(d))$ Pareto dominates $(\tau_2(0),
\tau_2(d))$. Additionally, the function $\tau_1 + \tau_2$ can be
represented by $(\tau_1(0)+\tau_2(0),
\tau_1(d)+\tau_2(d))$. Similarly, with \(\tau'_1(x)=2a_1x\) and
\(\tau'_2(x)=2a_2x\) we have \(\tau'_1\le\tau'_2\) if and only if
\(a_1\leq a_2\).  Addition works again as expected.

To generalize this concept, consider a class of functions $\mathcal T$
that is closed under addition.  We say that $\mathcal T$ has
\emph{Pareto dimension} $k$ if the following holds.  There exists a function $p\colon \mathcal T \to \mathbb Q^k$ such that $\tau_1$ dominates $\tau_2$ if and only if $p(\tau_1)$ Pareto dominates $p(\tau_2)$, and such that $p(\tau_1+\tau_2) = p(\tau_1)+p(\tau_2)$. We call $p$ the \emph{Pareto representation} of $\mathcal T$. The
above canonical cost functions have Pareto dimension~2 and their
derivatives have Pareto dimension~1.  

With this, $P_1 \preceq P_2$ reduces to having
$p(\tau_{P_1}) \le p(\tau_{P_2})$ and
$p'(\tau_{P_1 \cap Q}') \le p'(\tau_{P_2 \cap Q}')$, where~$p'$ is a
Pareto representation of the class of all derivatives of functions in
$\mathcal T$.  This is equivalent to the concatenation of
$p(\tau_{P_1})$ and $p'(\tau'_{P_1 \cap Q})$ Pareto dominating the
concatenation of $p(\tau_{P_2})$ and $p'(\tau'_{P_2 \cap Q})$.  Thus,
dominance of paths reduces to Pareto dominance.

\begin{restatable}{theorem}{sap}\label{thm:sap}
  SAP with Pareto-conform psychological model and cost functions with
  Pareto dimension \(k\) whose derivatives have Pareto dimension
  \(\ell\) reduces to solving a multi-criteria shortest path problem
  with $k+\ell$ criteria and scoring the result.
\end{restatable}
\begin{proof}
  Let $p$ and $p'$ denote the Pareto mappings for the cost functions and for
  their derivatives, respectively. Every edge $e=(v_i, v_j)$ is mapped to the
  $k+\ell$-dimensional vector obtained by concatenating $p(\tau_e)$
  and $p'(\tau_{e\cap Q}')$.  Because
  \(p(\tau_1+\tau_2) = p(\tau_1)+p(\tau_2)\) and
  $p'(\tau_1' + \tau_2') = p'(\tau_1') + p'(\tau_2')$, solving the
  $k+\ell$-criteria shortest path problem from $s$ to $t$ with these
  vectors yields all $s$-$t$-paths that are not dominated. By
  \cref{lem:sap-optimal-is-not-dominated}, the optimal solution is
  among these paths and can thus be obtained by computing $\Tcost_P$
  for every path $P$ in the result and choosing the best one.
\end{proof}

\subsection{Enforcing 1-Disjoint Routes}\label{subsec:1d_SAP}

1D-SAP can be solved by modifying the graph
and then applying the same approach as above.
Let $Q = (v_1, \dots, v_q)$ with $s = v_1$, $t = v_q$. We consider
the graph \(G'\), which is a copy of~\(G\) where for each
\(v_i\in\{v_2,\ldots,v_{q-1}\}\) a node \(v_i'\) is added.  Moreover,
\(v_i\) in~\(G'\) has all outgoing edges of \(v_i\) in \(G\), but only
the incoming edge from \(v_{i-1}\). Similarly, \(v_i'\) in \(G'\) has
all incoming edges of \(v_i\) in \(G\), but only the outgoing edge to
\(v_{i+1}'\); see~\Cref{fig:1d-sap}.
With this, computing all non-dominated 1-disjoint paths in $G$ reduces
to computing all non-dominated paths in $G'$.

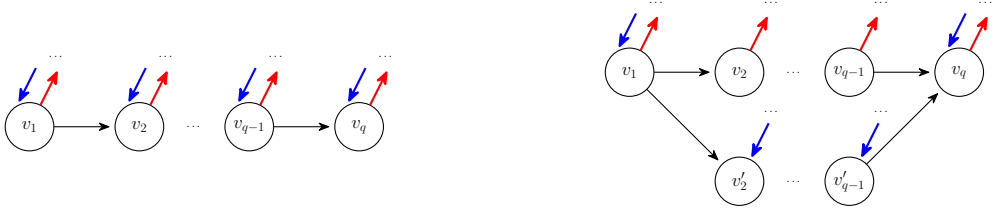
\begin{figure}
  \centering
  \hfill
  \begin{subfigure}{0.45\textwidth}
    \tikzexternalenable
\begin{tikzpicture}[->,>={Stealth[round,sep]},shorten >=1pt,auto,node distance=2cm,scale=0.4, every node/.style={transform shape},every path/.style={transform shape}]
  \begin{scope}[every node/.style={circle,draw=black,minimum size=1.6cm}]
    \node (1)              {\LARGE $v_1$};
    \node (2) [right =of 1]{\LARGE $v_2$};
    \node (3) [right =of 2]{\LARGE $v_{q-1}$};
    \node (4) [right =of 3]{\LARGE $v_{q}$};
\end{scope}

\begin{scope}[every edge/.style={draw=red,thick}, every path/.style={draw=red,thick}]
    \path[->] (2) edge node[above =of 2] {} +(1,2);
    \path[->] (3) edge node[above =of 3] {} +(1,2);
     \path[->] (1) edge node[above =of 1] {} +(1,2);
      \path[->] (4) edge node[above =of 4] {} +(1,2);
  \end{scope}
  
  \begin{scope}[every edge/.style={draw=blue,thick}, every path/.style={draw=blue,thick}]
   \path (4) edge [<-,transform canvas={xshift=-3mm, yshift=0.1mm}] node[above =of 4] {} +(1,2);
   \path (1) edge [<-,transform canvas={xshift=-3mm, yshift=0.1mm}] node[above =of 1] {} +(1,2);
     \path (2) edge [<-,transform canvas={xshift=-3mm, yshift=0.1mm}] node[above =of 2] {} +(1,2);
        \path (3) edge [<-,transform canvas={xshift=-3mm, yshift=0.1mm}] node[above =of 3] {} +(1,2);
  \end{scope}
  
    \draw (1) edge  (2);
    \draw (3) edge (4);
    \path (2) -- node[auto=false]{\ldots} (3);
    \path (2) -- node[auto=false]{\ldots} +(1.5,3.9);
    \path (3) -- node[auto=false]{\ldots} +(1.5,3.9);
     \path (1) -- node[auto=false]{\ldots} +(1.5,3.9);
      \path (4) -- node[auto=false]{\ldots} +(1.5,3.9);

\end{tikzpicture}
\tikzexternaldisable
  \end{subfigure}\hfill
  \begin{subfigure}{0.45\textwidth}
    \tikzexternalenable
\begin{tikzpicture}[->,>={Stealth[round,sep]},shorten >=1pt,auto,node distance=2cm,scale=0.4, every node/.style={transform shape},every path/.style={transform shape}]
  \begin{scope}[every node/.style={circle,draw=black,minimum size=1.6cm}]
    \node (1)              {\LARGE $v_1$};
    \node (2) [right =of 1]{\LARGE $v_2$};
    \node (3) [right =of 2]{\LARGE $v_{q-1}$};
    \node (4) [right =of 3]{\LARGE $v_{q}$};
    \node (5) [below =of 2]{\LARGE $v_2'$};
    \node (6) [right =of 5]{\LARGE $v_{q-1}'$};
\end{scope}
\node (0) at (4) {};
\begin{scope}[every edge/.style={draw=red,thick}, every path/.style={draw=red,thick}]
    \path[->] (2) edge node[above =of 2] {} +(1,2);
    \path[->] (3) edge node[above =of 3] {} +(1,2);
     \path[->] (1) edge node[above =of 1] {} +(1,2);
      \path[->] (4) edge node[above =of 4] {} +(1,2);
  \end{scope}
  
  \begin{scope}[every edge/.style={draw=blue,thick}, every path/.style={draw=blue,thick}]
   \path[<-] (6) edge node[above =of 6] {} +(1,2);
   \path[<-] (5) edge node[above =of 5] {} +(1,2);
   \path (4) edge [<-,transform canvas={xshift=-3mm, yshift=0.1mm}] node[above =of 4] {} +(1,2);
    \path (1) edge [<-,transform canvas={xshift=-3mm, yshift=0.1mm}] node[above =of 1] {} +(1,2);
  \end{scope}
  
    \draw (1) edge  (2);
    \draw (1) edge  (5);
    \draw (3) edge (4);
    \draw (6) edge (4);
    \path (2) -- node[auto=false]{\ldots} (3);
    \path (5) -- node[auto=false]{\ldots} (6);
    \path (2) -- node[auto=false]{\ldots} +(1.8,3.9);
    \path (3) -- node[auto=false]{\ldots} +(1.8,3.9);
     \path (1) -- node[auto=false]{\ldots} +(1.5,3.9);
     \path (5) -- node[auto=false]{\ldots} +(1.8,3.9);
      \path (6) -- node[auto=false]{\ldots} +(1.8,3.9);
      \path (4) -- node[auto=false]{\ldots} +(1.5,3.9);

\end{tikzpicture}
\tikzexternaldisable
  \end{subfigure}
  \hfill
  \caption{Graph transformation for the 1D-SAP algorithm. Blue edges represent an arbitrary number of incoming edges, red edges an arbitrary number of outgoing edges.}\label{fig:1d-sap}
\end{figure}

\begin{restatable}{theorem}{1d-sap-by-sap}\label{thm:1d-sap-by-sap}
    \Cref{thm:sap} also holds for 1D-SAP.
\end{restatable}
\begin{proof}
  We show that for every 1-disjoint $st$-path \(P\) in \(G\) there is
  an $st$-path $P'$ in \(G'\) with the same cost function and vice
  versa.
  With this we can apply the same approach as in \Cref{thm:sap}
  to compute the set of non-dominated \(st\)-paths in \(G'\), which represents
  the set of non-dominated 1-disjoint \(st\)-paths in \(G\). Afterwards, we
  compute the optimal solution in $G$ as in \Cref{thm:sap}.
  Let $P'$ be an $st$-path in $G'$ and let
  $Q_{\mathrm{out}} = (s, v_2, \dots, v_{q-1})$ and
  $Q_{\mathrm{in}} = (v_2', \dots, v_{q-1}', t)$ be the parts of the
  original route and its copy in the modified graph~$G'$ that have the
  outgoing and incoming edges augmented with $s$ and $t$,
  respectively.  First observe, that $P' \cap Q_{\mathrm{out}}$ is a
  prefix of $P'$, as $P'$ cannot get back to $Q_{\mathrm{out}}$ after
  it diverted from it due to the fact that we removed all incoming
  edges from vertices in $Q_{\mathrm{out}}$ (other than $s$).
  Similarly, $P' \cap Q_{\mathrm{in}}$ is a suffix of $P'$ as $t$ is
  the only vertex on $Q_{\mathrm{in}}$ with outgoing edges.  Thus, the
  path $P$ in $G$ we get by replacing every $v_i'$ in $P'$ with the
  corresponding $v_i$ is 1-disjoint and has the same cost function as
  $P'$.  Analogously, every 1-disjoint $st$-path $P$ in $G$ can be
  transformed to such a path $P'$ in $G'$ with the same cost.
\end{proof}

\subsection{Fewer Criteria for Disjoint SAP}\label{sec:disjoint-single-alt-path}

We now consider the D-SAP variant whose major advantage is
that we can solve it with fewer criteria in the multi-criteria
shortest path part of the algorithm.
Analogously to \Cref{lem:sap-optimal-is-not-dominated}, the following lemma
follows from \Cref{def:pareto-conformity}.

\begin{lemma}\label{lem:dsap-optimal-is-not-dominated}
	For any instance of D-SAP with a weakly Pareto-conform psychological model,
	there exists an optimal solution that is not dominated by any other alternative disjoint from~$Q$.
\end{lemma}

To guarantee that we only find paths disjoint from $Q$, we remove $Q$
from the graph.  For two paths $P_1$ and $P_2$ in the resulting graph,
the dominance $P_1 \preceq P_2$ simplifies to
$\tau_{P_1}\le \tau_{P_2}$.
This observation together with
\Cref{lem:dsap-optimal-is-not-dominated} gives us the following
theorem.  Note that we only need weak Pareto-conformity here, as all
paths have no intersection with $Q$.

\begin{theorem}
  D-SAP with a weakly Pareto-conform psychological model and cost
  functions with Pareto dimension $k$ reduces to solving a
  multi-criteria shortest path problem with $k$ criteria and scoring
  the result.
\end{theorem}

\subsection{Fewer Criteria for 1-Disjoint SAP}\label{sec:1d-sap-dp}
We start by deleting the edges of $Q = (v_1, \dots, v_q)$ from the graph. In the resulting
graph, for every pair $1 \le i < j \le q$, we calculate the set
$\mathcal P_{i, j}$ of all routes between $v_i$ and $v_j$ that are
minimal with respect to dominance. We define the corresponding
\emph{augmented path} for a path \(P \in \mathcal P_{i, j}\) as
$\widetilde P = (v_1, \dots, v_{i - 1}) \cup P \cup (v_{j + 1} \dots,
v_q)$, which is a \emph\emph{1-disjoint} path from $s$ to $t$.  We
denote the set of paths from~$s$ to $t$ obtained by augmenting all
paths in $\mathcal P_{i, j}$ by $\widetilde {\mathcal P}_{i, j}$.

\begin{restatable}{lemma}{fewer-crit-1-disj}
  For an instance of 1D-SAP with weakly Pareto-conform psychological model,
  there exists an optimal solution among the paths in the sets
  $\widetilde {\mathcal P}_{i, j}$.
\end{restatable}
\begin{proof} Let $\widetilde P$ be an optimal alternative path and let $v_i$ and
  $v_j$ be the vertices where $\widetilde P$ diverts and rejoins $Q$
  respectively.  Moreover, let $P$ be the subpath of $\widetilde P$
  from $v_i$ to $v_j$.  If $P$ is in $\mathcal P_{i, j}$, then
  $\widetilde P$ is also in $\widetilde {\mathcal P}_{i, j}$ and the
  statement is true.  Otherwise, if $P$ is not in $\mathcal P_{i, j}$,
  there is a different path $P'$ from $v_i$ to $v_j$ that dominates
  $P$.  Then $\widetilde P'$ also dominates $\widetilde P$, as in both
  cases we augment the paths with the same edges, which equally
  increases their cost functions and ensures
  $\widetilde P' \cap Q = \widetilde P \cap Q$.  Thus, by weak
  Pareto-conformity (\Cref{def:pareto-conformity}), the path
  $\widetilde P'$ is not worse than $\widetilde P$ and thus
  $\widetilde P'$ is an optimal alternative route that is contained in
  $\widetilde {\mathcal P}_{i, j}$.
\end{proof}

From~\Cref{sec:disjoint-single-alt-path}, we know that we can compute
all non-dominated paths from $v_i$ to $v_j$ by using a multi-criteria
shortest path algorithm.  Thus, 1D-SAP reduces to solving
$\binom{q}{2} \in \Theta(q^2)$ multi-criteria shortest path problems, one
for each pair of vertices $v_i, v_j \in Q$.  We note that many
shortest path algorithms actually solve a more general problem by
computing paths from a single start to all other vertices. Thus,
instead of $\Theta(q^2)$ shortest path problems, we can solve $q$
multi-target shortest path problems, using each vertex in $Q$ as start
once.

\begin{theorem}
  1D-SAP with weakly Pareto-conform psychological model and cost
  functions with Pareto dimension $k$ reduces to solving $q$
  multi-criteria multi-target shortest path problems with $k$ criteria
  and scoring the resulting augmented paths.
\end{theorem}

\subsection{Fewer Criteria for SAP}\label{sec:dynamic-program-sap}
We now provide an algorithm for the SAP problem that requires fewer
criteria. We use a dynamic program that combines non-dominated subpaths to
obtain the optimal solution. To formalize this, we need
the following additional notation.  For $v_i, v_j \in Q$ with $i < j$, a path
from
$v_i$ to $v_j$ is called \emph{$Q$-path} or more specifically
\emph{$Q_{i,j}$-path}.  A set $A$ of $Q_{i,j}$-paths is \emph{reduced}
if no path in $A$ is dominated by another path in $A$.  Let $A$
and $B$ be two reduced sets of $Q_{i,j}$-paths.  Their \emph{reduced
  union} is obtained by eliminating from $A \cup B$ all paths that are dominated by another path in $A \cup B$.  Moreover, let $A$ and
$B$ be two reduced sets of $Q_{i,j}$ and $Q_{j, k}$-paths,
respectively.  Then their \emph{reduced join} is obtained by
concatenating every path in $A$ with every path in $B$ and
eliminating all dominated paths.

We start by applying the algorithm from \cref{sec:1d-sap-dp}, computing the sets \(\mathcal{P}_{i,j}\) for all
\(1\leq i \leq j\leq q\), which are the reduced sets of
all $Q_{i,j}$-paths that are disjoint from \(Q\). Then, we compute sets $\mathcal
P_j$ of $Q_{1, j}$-paths and one can show that $\mathcal P_j$ is in fact the reduced set of all $Q_{1,j}$-paths.  We initialize
$\mathcal P_1 = \set{(v_1)}$.  Now, assume we have computed
$\mathcal P_i$ for all $i < j$. We obtain $\mathcal P_j$ as the
reduced union of the following sets of $Q_{1, j}$-paths: the reduced
join of $\mathcal P_i$ and $\mathcal P_{i, j}$ for every $i < j$, and
the reduced join of $\mathcal P_{j - 1}$ and $\{(v_{j-1}, v_j)\}$.

\begin{restatable}{lemma}{fewer-crit-sap-dp}
    For an instance of SAP with a Pareto-conform psychological model,
    there exists an optimal solution among the paths in \(\mathcal P_q\).
\end{restatable}
\begin{proof}
  We show for every \(1 \le j \le q\) that all paths from \(v_1\) to
  \(v_j\) that are not dominated by any other path from \(v_1\) to
  \(v_j\) are contained in \(\mathcal{P}_j\).

  We prove by induction. For \(j = 1\), we have
  \(\mathcal{P}_1 = \set{(v_1)}\), which is the set of all
  non-dominated \(v_1,v_1\)-paths. Now, assume that for all
  \(1 \le i < j\), \(\mathcal{P}_i\) is the set of non-dominated paths
  from \(v_1\) to \(v_i\). Let \(P\) be a path from \(v_1\) to \(v_j\)
  that is not dominated by any other path from \(v_1\) to \(v_j\). We
  consider two cases.

  Case 1: \((v_{j-1}, v_j) \in P\). Let
  \(P_{j-1} = P \setminus \set{(v_{i-1},v_i)}\).  If
  \(P_{j-1} \in \mathcal{P}_{j-1}\), then, we also have
  \(P \in \mathcal{P}_j\) and there is nothing to show.  Otherwise, if
  \(P_{j-1} \notin \mathcal{P}_{j-1}\), there exists a
  \(P_{j-1}' \in \mathcal{P}_{j-1}\) with
  \(P_{j-1}' \preceq P_{j-1}\), which implies
  \(P_{j-1}' \cup \set{(v_{j-1},v_j)} \preceq P_{j-1} \cup
  \set{(v_{j-1},v_j)} = P\). This is a contradiction to our assumption
  that \(P\) is not dominated by any other path from \(v_1\) to
  \(v_j\).

  Case 2: \((v_{j-1}, v_j) \notin P\).  Let $v_i$ be the last vertex
  $P$ shares with $Q$, i.e., $1 \le i < j$ is the maximum such that
  $v_i$ lies on $P$.  Then $P$ consists of two subpaths
  \(P = P_1 \cup P_2\) such that $P_1$ goes from \(v_1\) to \(v_i\)
  and \(P_2\) from \(v_i\) to \(v_j\).  Note that $P_2$ is disjoint from
  $Q$.  If \(P_1 \in \mathcal{P}_i\) and
  \(P_2 \in \mathcal{P}_{i,j}\), then we also have
  \(P \in \mathcal{P}_j\) and there is nothing to show.  Thus, it
  remains to deal with the cases that \(P_1 \notin \mathcal{P}_i\) or
  \(P_2 \notin \mathcal{P}_{i,j}\).

  In the first case, for some path \(P_1' \in \mathcal{P}_i\), we have \(P_1'
  \preceq P_1\). Since \(\tau_{P_1'} \le \tau_{P_1}\) implies
  \(\tau_{P_1'\cup P_2} \le \tau_{P_1 \cup P_2}\) and since
  \(\tau_{(P_1' \cup P_2)\cap Q}' = \tau_{P_1'\cap Q}' \le
  \tau_{P_1\cap Q}' = \tau_{(P_1 \cup P_2)\cap Q}'\), we have \(P_1'
  \cup P_2 \preceq P_1 \cup P_2 = P\). This is a contradiction to our
  assumption that \(P\) is not dominated.

  In the second case, for some path \(P_2' \in \mathcal{P}_{i,j}\), we have
  \(P_2' \preceq P_2\). Since \(\tau_{P_2'} \le \tau_{P_2}\) implies
  \(\tau_{P_1\cup P_2'} \le \tau_{P_1 \cup P_2}\) and since
  \(\tau_{(P_1 \cup P_2)\cap Q}' = \tau_{(P_1 \cup P_2')\cap Q}'\), we
  have \(P_1 \cup P_2' \preceq P_1 \cup P_2 = P\). This is a
  contradiction to our assumption that \(P\) is not dominated.
\end{proof}

After computing the sets \(\mathcal{P}_{i, j}\) as in
\cref{sec:1d-sap-dp}, it remains to compute \(\Theta(q^2)\) reduced
joins and \(\Theta(q^2)\) reduced unions between reduced sets of
\(Q\)-paths. Then, it only remains to score the result
\(\mathcal{P}_q\). The shortest path computations from
\cref{sec:1d-sap-dp} use \(k\) criteria where \(k\) is the
Pareto-dimension of the cost functions.  The reduced joins and unions
are with respect to \(k+\ell\) criteria, where \(\ell\) is the Pareto
dimension of the derivatives of the cost functions.

\begin{theorem}
  SAP with Pareto-conform psychological model and cost functions with
  Pareto dimension \(k\) whose derivatives have Pareto dimension
  \(\ell\) reduces to solving \(q\) multi-target
  shortest path problems with \(k\) parameters, executing
  \(\Theta(q^2)\) reduced join and union operations with respect to
  \(k+\ell\) criteria between reduced sets of \(Q\)-paths, and scoring
  the result \(\mathcal{P}_{q}\).
\end{theorem}

\section{Psychological Models and Pareto-Conformity}\label{sec:psychmods}

In this section, we formally define the models mentioned in
\cref{subsec:control_rate} and show their Pareto-conformity.  After
considering the System Optimum Model, we define the Quotient Model,
which is a generalization of the User Equilibrium
Model and the Linear Model. We give conditions under
which a Quotient Model is Pareto-conform and thereby prove that the
User Equilibrium Model and the Linear Model are both Pareto-conform.


The \emph{System Optimum Model} assumes that agents distribute
optimally, i.e., $x_P \in [0, d]$ minimizes $\Tcost_P(x_P)$.  We get
that $P_1 \preceq P_2$ implies
$\Tcost_{P_1}(x) \le \Tcost_{P_2}(x)$ for each $x \in [0, d]$.

\begin{restatable}{theorem}{so-pareto-conf}
    The System Optimum Model is Pareto-conform.
\end{restatable}
\begin{proof}
  Let two alternative paths $P_1,P_2$ and an original path $Q$ be
  given with $P_1\preceq P_2$, i.e., $\tau_{P_1}\leq\tau_{P_2}$ and
  $\tau'_{P_1\cap Q}\leq\tau'_{P_2\cap Q}$.  The core of the proof
  is to show that \(\Tcost_{P_1}(x) \leq \Tcost_{P_2}(x)\) for all
  \(x\in[0, d]\).  Then, for the minima $x_{P_1}, x_{P_2} \in [0, d]$
  of $\Tcost_{P_1}(x)$ and $\Tcost_{P_2}(x)$, respectively, we get
  \(\Tcost_{P_1}=\Tcost_{P_1}(x_{P_1})\leq\Tcost_{P_1}(x_{P_2})\leq
  \Tcost_{P_2}(x_{P_2})=\Tcost_{P_2}\), which proves
  Pareto-conformity.

  It remains to show \(\Tcost_{P_1}(x) \leq \Tcost_{P_2}(x)\) for all
  \(x\in[0, d]\).  Starting with \(\Tcost_{P}(x)\) as given in~\Cref{eq:cost-P}, we obtain
  \begin{align*}
    \Tcost_{P}(x)
    &= x\cdot \tau_{P\setminus Q}(x) + (d-x)\cdot\tau_{Q\setminus
      P}(d-x)+d\cdot\tau_{P \cap Q} (d)\\
    &= x\cdot \left(\tau_{P\setminus Q}(x) + \tau_{P \cap Q}(d)\right)
      + (d - x)\cdot\left(\tau_{Q\setminus P}(d-x) + \tau_{P \cap
      Q}(d)\right).
  \end{align*}
  As $\tau_P(x) = \tau_{P\setminus Q}(x) + \tau_{P \cap Q}(x)$ and
  $\tau_Q(d - x) = \tau_{Q\setminus P}(d - x) + \tau_{P \cap Q}(d -
  x)$ this yields
  \begin{equation*}
    \Tcost_{P}(x) = x\big(\underbracket{\tau_{P}(x)}_{(i)} + \underbracket{\tau_{P \cap Q}(d) - \tau_{P \cap Q}(x)}_{(ii)}\big) + (d - x)\big(\tau_{Q}(d-x) + \underbracket{\tau_{P \cap Q}(d) - \tau_{P \cap Q}(d - x)}_{(iii)}\big).
  \end{equation*}
  When going from $\Tcost_{P_1}$ to $\Tcost_{P_2}$, the three parts
  $(i)$, $(ii)$ and $(iii)$ change.  As $P_1 \preceq P_2$, we directly
  get \(\tau_{P_1}(x) \leq \tau_{P_2}(x)\) for $(i)$.  Similarly,
  \(\tau_{P_1\cap Q}'\leq \tau_{P_2\cap Q}'\), implies that $(ii)$ and
  $(iii)$ for $P_1$ are upper-bounded by the corresponding terms for
  $P_2$.  Thus, $\Tcost_{P_1}(x) \le \Tcost_{P_2}(x)$.
\end{proof}


\noindent For the \emph{Quotient Model}, let $c(x)$ be
non-decreasing, non-negative on $[0, d]$ with \(c(d)>0\).
If
\begin{equation}
  \label{eq:general-model-def}
  \frac{\tau_{Q\setminus P}(d - x) + \tau_{P\cap
      Q}(d)}{\tau_{P\setminus Q}(x) + \tau_{P\cap Q}(d)} = c(x)
\end{equation}
has a solution in $[0, d]$, it is unique for the following reason.
The numerator and denominator are the cost of $Q$ and $P$, which are
decreasing and increasing in $x$, respectively.  Thus, the
quotient is decreasing, while $c(x)$ is non-decreasing, which makes
the solution unique.  The Quotient Model sets $x_P$ to this unique
solution if it exists.
If no solution exists, then the left-hand side is either smaller or
larger than $c(x)$ for every $x \in [0, d]$, in which case we set
$x_P = 0$ or $x_P = d$, respectively.  This is the natural choice, as
$x_P = 0$ and $x_P = d$ maximizes and minimizes the left-hand side,
respectively.
We note that $c$ specifies how conservative the agents are.  If
$c(x) = 1$, the agents distribute on $Q$ and $P$ such that both paths
have the same cost.  If $c$ is smaller, then agents take the
alternative route, if it is not too much longer.

Recall from~\Cref{eq:cost-P} that the cost function
$\Tcost_P(x)$ is a combination of the three functions
$\tau_{P \setminus Q}$, $\tau_{Q \setminus P}$, and
$\tau_{P \cap Q}$.  If~\Cref{eq:general-model-def} has a
solution $x_P$, we know how $\tau_{P \setminus Q}(x)$ and
$\tau_{Q \setminus P}(x)$ relate to each other at $x = x_P$.  In other
words, solving~\Cref{eq:general-model-def} for
$\tau_{P \setminus Q}(x_P)$ or $\tau_{Q \setminus P}(x_P)$ and
replacing their occurrence in $\Tcost_P = \Tcost_P(x_P)$ with the
result lets us eliminate $\tau_{Q \setminus P}$ or
$\tau_{P \setminus Q}$, respectively, from $\Tcost_P$.  We do this in
the following two lemmas, which additionally take the special cases
$x_P = 0$ and $x_P = d$ into account.

\begin{restatable}{lemma}{general-model-Cp-with-P}\label{lem:general-model-Cp-with-P}
  Let $g_P(x) = (d - x)\cdot c(x) + x$.  Then
  $\Tcost_P \le g_P(x_P) \cdot \left(\tau_{P \setminus Q}(x_P) +
    \tau_{P \cap Q}(d)\right)$.  If $x_P > 0$, then equality holds.
\end{restatable}
\begin{proof}
  First assume $x_P = 0$.  Recall that in this case the left-hand side
  of~\Cref{eq:general-model-def} is at most its right-hand
  side.  Thus, for $x_P = 0$ we obtain
  \begin{equation*}
    \Tcost_P = \Tcost_P(0) = d\cdot \left(\tau_{Q\setminus P}(d) + \tau_{Q\cap P}(d)\right)
    \le d\cdot c(0)\cdot (\tau_{P \setminus Q}(0) + \tau_{Q\cap P}(d)),
  \end{equation*}
  which proves the claim for $x_P = 0$.
  For all other cases, we have to show equality.  First assume
  $x_P = d$.  Then the claim simplifies to
  $\Tcost_P = d\cdot (\tau_{P\setminus Q}(d) + \tau_{P\cap Q}(d))$,
  which is true.

  It remains to consider $0 < x_P < d$, in which case $x_P$ is a
  solution of~\Cref{eq:general-model-def}.  Solving~\Cref{eq:general-model-def} for
  $\tau_{Q\setminus P}(d - x_P)$ gives
  \begin{equation*}
    \tau_{Q\setminus P}(d - x_P)
    = c(x_P)\cdot \tau_{P \setminus Q}(x_P) + (c(x_P) - 1) \cdot \tau_{P \cap Q}(d),
  \end{equation*}
  and plugging it into the cost function $\Tcost_P = \Tcost_P(x_P)$ in~\Cref{eq:cost-P} yields
  \begin{align*}
    \Tcost_P
    &= x_P\cdot \tau_{P\setminus Q}(x_P) + (d - x_P) \cdot \tau_{Q
      \setminus P}(d - x_p) + d \cdot \tau_{P \cap Q}(d)\\
    &= (x_P + (d - x_P)\cdot c(x_P))\cdot \tau_{P\setminus Q}(x_P) +
      (d + (d - x_P)\cdot (c(x_P) - 1))\cdot \tau_{P \cap Q}(d)\\
    &= ((d - x_P)\cdot c(x_P) + x_P)\cdot \left(\tau_{P\setminus
      Q}(x_P) + \tau_{P \cap Q}(d)\right),
  \end{align*}
  which proves the claim.
\end{proof}

We note that \(c(x_P)>0\) holds for the following reason. For
\(x_P=d\) this is true by definition. For $x_P < d$, the left-hand
side of~\Cref{eq:general-model-def} is equal to its right-hand side or
less (in which case \(x_P=0\)).  As the right-hand side is \(c(x_P)\)
and the left-hand side is positive, we get $c(x_P) > 0$. Thus it is
fine to divide by \(c(x_P)\) in the following lemma.

\begin{restatable}{lemma}{general-model-Cp-with-Q}\label{lem:general-model-Cp-with-Q}
  Let $g_Q(x) = d + x / c(x) - x$.  Then
  $\Tcost_P \le g_Q(x_P) \cdot \left(\tau_{Q \setminus P}(d - x_P) +
    \tau_{P \cap Q}(d)\right)$. If $x_P < d$, then equality holds.
\end{restatable}
\begin{proof}
  The proof is similar to that of~\Cref{lem:general-model-Cp-with-P}.  First, if $x_P = d$, then
  the right-hand side of~\Cref{eq:general-model-def} is at
  most its left-hand side.  Thus, for $x_P = d$, we obtain
  \begin{equation*}
    \Tcost_P = \Tcost_P(d) = d\cdot\left(\tau_{P\setminus Q}(d) +
      \tau_{P\cap Q}(d)\right)
    \le \frac{d}{c(d)}\cdot \left(\tau_{Q\setminus P}(0) + \tau_{P\cap
      Q}(d)\right),
  \end{equation*}
  which proves the claim for $x_P = d$.
  For all other cases, we have to show equality. Now, first assume
  $x_P = 0$.  Then the claim simplifies to
  $\Tcost_P = d\cdot(\tau_{Q\setminus P}(d) + \tau_{P\cap Q}(d))$,
  which is true.

  It remains to consider $0 < x_P < d$, in which case $x_P$ is a
  solution of~\Cref{eq:general-model-def}.  Solving~\Cref{eq:general-model-def} for $\tau_{P\setminus Q}(x_P)$
  gives
  \begin{equation*}
    \tau_{P\setminus Q}(x_P) = \frac{1}{c(x_P)} \cdot \tau_{Q\setminus
      P}(d - x_P) + \left(\frac{1}{c(x_P)} - 1\right) \cdot \tau_{P \cap Q}(d),
  \end{equation*}
  and plugging it into the cost function $\Tcost_P = \Tcost_P(x_P)$ in~\Cref{eq:cost-P} yields
  \begin{align*}
    \Tcost_P
    &= x_P\cdot \tau_{P\setminus Q}(x_P) + (d - x_P) \cdot \tau_{Q
      \setminus P}(d - x_P) + d \cdot \tau_{P \cap Q}(d)\\
    &= \left(d - x_P + x_P \cdot \frac{1}{c(x_P)}\right) \cdot \tau_{Q
      \setminus P}(d - x_P) + \left(d + x_P \cdot
      \left(\frac{1}{c(x_P)} - 1\right)\right) \cdot \tau_{P \cap
      Q}(d)\\
    &= \left(d + \frac{x_P}{c(x_P)} - x_P\right) \cdot \left(\tau_{Q
      \setminus P}(d - x_P) + \tau_{P \cap Q}(d)\right),
  \end{align*}
  which proves the claim.
\end{proof}

The following lemma provides the core inequalities we need when
comparing the cost of two alternative paths.  Note how the
inequalities in parts~\ref{itm:psy-model-core-inequalities-1-le-2} and~\ref{itm:psy-model-core-inequalities-1-ge-2} of the lemma resemble the
representation of the cost $\Tcost_P$ in~\Cref{lem:general-model-Cp-with-P} and~\Cref{lem:general-model-Cp-with-Q}, respectively.

\begin{restatable}{lemma}{psy-model-core-inequalities}\label{lem:psy-model-core-inequalities}
  Let $P_1$ and $P_2$ be alternative paths with $P_1 \preceq P_2$ and
  let $x_1, x_2 \in [0, d]$.  Then
  \begin{enumerate}
  \item\label{itm:psy-model-core-inequalities-1-le-2}
    $\tau_{P_1 \setminus Q}(x_1) + \tau_{P_1 \cap Q}(d) \le
    \tau_{P_2 \setminus Q}(x_2) + \tau_{P_2 \cap Q}(d)$ if
    $x_1 \le x_2$, and
  \item\label{itm:psy-model-core-inequalities-1-ge-2}
    $\tau_{Q\setminus P_1}(d - x_1) + \tau_{P_1 \cap Q}(d) \le
    \tau_{Q\setminus P_2}(d - x_2) + \tau_{P_2 \cap Q}(d)$ if
    $x_1 \ge x_2$.
  \end{enumerate}
\end{restatable}
\begin{proof}
  Recall that $P_1 \preceq P_2$ means that for all $x \in [0, d]$,
  $\tau_{P_1}(x) \le \tau_{P_2}(x)$ and
  $\tau_{P_1 \cap Q}'(x) \le \tau_{P_2 \cap Q}'(x)$, where $\tau'$
  denotes the derivative of $\tau$.  For the first case $x_1 \le x_2$
  we get
  \begin{align*}
    \tau_{P_1 \setminus Q}(x_1) + \tau_{P_1 \cap Q}(d)
    &= \tau_{P_1 \setminus Q}(x_1) + \tau_{P_1 \cap Q}(x_1) -
      \tau_{P_1 \cap Q}(x_1) + \tau_{P_1 \cap Q}(d)\\
    &= \tau_{P_1}(x_1) + \tau_{P_1 \cap Q}(d) - \tau_{P_1 \cap Q}(x_1),
      \intertext{using that $\tau_{P_1}(x) \le \tau_{P_2}(x)$ and
      $\tau_{P_1 \cap Q}'(x) \le \tau_{P_2 \cap Q}'(x)$}
    &\le \tau_{P_2}(x_1) + \tau_{P_2 \cap Q}(d) - \tau_{P_2 \cap Q}(x_1)\\
    &= \tau_{P_2 \setminus Q}(x_1) + \tau_{P_2 \cap Q}(x_1) +
      \tau_{P_2 \cap Q}(d) - \tau_{P_2 \cap Q}(x_1)\\
    &= \tau_{P_2 \setminus Q}(x_1) + \tau_{P_2 \cap Q}(d).
  \end{align*}
  As $\tau_{P_2 \setminus Q}$ is an increasing function and
  $x_2 \ge x_1$, we obtain
  $\tau_{P_2 \setminus Q}(x_1) \le \tau_{P_2 \setminus Q}(x_2)$, which
  concludes this case.

  The case $x_1 \ge x_2$ works very similar.  We obtain
  \begin{align*}
    \tau_{Q \setminus P_1}(d - x_1) + \tau_{P_1 \cap Q}(d)
    &= \tau_{Q \setminus P_1}(d - x_1) + \tau_{P_1 \cap Q}(d - x_1) -
      \tau_{P_1 \cap Q}(d - x_1) + \tau_{P_1 \cap Q}(d)\\
    &= \tau_{Q}(d - x_1) + \tau_{P_1 \cap Q}(d) - \tau_{P_1 \cap Q}(d -x_1),
      \intertext{and using that $\tau_{P_1 \cap Q}'(x) \le \tau_{P_2
      \cap Q}'(x)$, we get}
    &\le \tau_{Q}(d - x_1) + \tau_{P_2 \cap Q}(d) - \tau_{P_2 \cap Q}(d -x_1)\\
    &= \tau_{Q \setminus P_2}(d - x_1) + \tau_{P_2 \cap Q}(d - x_1) +
      \tau_{P_2 \cap Q}(d) - \tau_{P_2 \cap Q}(d -x_1)\\
    &= \tau_{Q \setminus P_2}(d - x_1) + \tau_{P_2 \cap Q}(d).
  \end{align*}
  As $\tau_{Q \setminus P_2}$ is an increasing function,
  $\tau_{Q \setminus P_2}(d - x)$ decreases in $x$.  Thus, as
  $x_1 \ge x_2$, we obtain
  $\tau_{Q \setminus P_2}(d - x_1) \le \tau_{Q \setminus P_2}(d -
  x_2)$, which concludes the proof.
\end{proof}
Applying the previous three lemmas and dealing with the additional
functions $g_P(x)$ and $g_Q(x)$ in~\Cref{lem:general-model-Cp-with-P} and~\Cref{lem:general-model-Cp-with-Q}, respectively, yields the
following.

\begin{restatable}{theorem}{quotient-model-pareto-conform}\label{thm:quotient-model-pareto-conform}
  The Quotient Model is Pareto-conform if $c(d) \le 1$ and, for all
  $x \in [0, d]$, $c(x)\cdot(1 - c(x)) - x \cdot c'(x) \le 0$.
\end{restatable}
\begin{proof}
  Consider two paths $P_1$ and $P_2$ such that $P_1 \preceq P_2$.  We
  want to show $\Tcost_{P_1} \le \Tcost_{P_2}$.

  First, assume the case $x_{P_1} \le x_{P_2}$.  If
  $x_{P_1} = x_{P_2} = 0$, we get $\Tcost_{P_1} = \Tcost_{P_2}$ as all
  traffic is routed over Q for both alternative paths.  Hence, assume
  $x_{P_1} \ge 0$ and $x_{P_2} > 0$.  By~\Cref{lem:general-model-Cp-with-P}, we get
  $\Tcost_{P_1} \le g_P(x_{P_1}) \cdot \left(\tau_{P_1 \setminus
      Q}(x_{P_1}) + \tau_{P_1 \cap Q}(d)\right)$, with
  $g_P(x) = (d - x)\cdot c(x) + x$.  For the second factor, we can use
  Lemma~\ref{lem:psy-model-core-inequalities}.\ref{itm:psy-model-core-inequalities-1-le-2}
  to obtain
  $\tau_{P_1 \setminus Q}(x_{P_1}) + \tau_{P_1 \cap Q}(d) \le
  \tau_{P_2 \setminus Q}(x_{P_2}) + \tau_{P_2 \cap Q}(d)$.  Concerning
  the first factor, $g_P(x)$ is non-decreasing if $g_P'(x) \ge 0$ for
  $x \in [0, d]$.  We get $g_P'(x) = (d - x) \cdot c'(x) - c(x) + 1$.
  As $c(x)$ is non-decreasing, $(d - x)\cdot c'(x)$ and $- c(x)$ have
  their minimum at $x = d$.  Thus, for $x \in [0, d]$,
  $g_P'(x) \ge g_P'(d) = 1 - c(d)$.  As $c(d) \le 1$ is required by
  the theorem, $g_P(x)$ is non-decreasing, which yields
  $g_P(x_{P_1}) \le g_P(x_{P_2})$.  To summarize, we thus get
  \begin{equation*}
    \Tcost_{P_1}
    \le g_P(x_{P_1}) \cdot \left(\tau_{P_1 \setminus Q}(x_{P_1}) +
      \tau_{P_1 \cap Q}(d)\right)
    \le g_P(x_{P_2}) \cdot \left(\tau_{P_2 \setminus Q}(x_{P_2}) +
      \tau_{P_2 \cap Q}(d)\right)
    = \Tcost_{P_2},
  \end{equation*}
  where the last equality holds due to~\Cref{lem:general-model-Cp-with-P} and the fact that
  $x_{P_2} > 0$.

  Secondly, consider the case $x_{P_1} > x_{P_2}$.  We
  use~\Cref{lem:general-model-Cp-with-Q} and hereby get
  $\Tcost_{P_1} \le g_Q(x_{P_1}) \cdot \left(\tau_{Q\setminus P_1}(d -
    x_{P_1}) + \tau_{P_1\cap Q}(d)\right)$, with
  $g(x) = d + x / c(x) - x$.  For the second factor, we can use~\Cref{lem:psy-model-core-inequalities}.\ref{itm:psy-model-core-inequalities-1-ge-2}
  to obtain
  $\tau_{Q\setminus P_1}(d - x_{P_1}) + \tau_{P_1\cap Q}(d) \le
  \tau_{Q\setminus P_2}(d - x_{P_2}) + \tau_{P_2\cap Q}(d)$.
  Concerning the first factor, $g_Q(x)$ is non-increasing if
  $g_Q'(x) \le 0$ for $x \in [0, d]$.  We get
  $g_Q'(x) = (c(x) - x\cdot c'(x))/c^2(x) - 1$, which is at most $0$
  if $c(x) - x\cdot c'(x) - c^2(x) \le 0$, which is required by the
  theorem.  To summarize, we thus get
  \begin{align*}
    \Tcost_{P_1}
    &\le g_Q(x_{P_1}) \cdot \left(\tau_{Q\setminus P_1}(d -
      x_{P_1}) + \tau_{P_1\cap Q}(d)\right)\\
    &\le g_Q(x_{P_2}) \cdot \left(\tau_{Q\setminus P_2}(d - x_{P_2}) +
      \tau_{P_2\cap Q}(d)\right) = \Tcost_{P_2},
  \end{align*}
  where the last equality holds due to~\Cref{lem:general-model-Cp-with-Q} and the fact that
  $x_{P_2} < x_{P_1} \le d$.
\end{proof}

As mentioned in in~\cref{subsec:control_rate}, the \emph{User Equilibrium Model}
sets $x_P$ such that both paths $P$ and $Q$ have the same cost per agent, if
possible.  More formally and in terms of the Quotient Model, we obtain the User
Equilibrium Model by setting $c(x) = 1$ in~\Cref{eq:general-model-def}, which is non-decreasing, non-negative, and satisfies $c(d) > 0$.

\begin{corollary}
  The User Equilibrium Model is Pareto-conform.
\end{corollary}
\begin{proof}
  To apply~\Cref{thm:quotient-model-pareto-conform}, we have to
  check whether the constant function $c(x) = 1$ satisfies the
  requirements.  Clearly $c(d) \le 1$.  Moreover, $c'(x) = 0$ and thus
  the left-hand side of the last condition resolves to $0$.
\end{proof}
The \emph{Linear Model} is defined by setting $c(x)$ to be an
increasing linear function, i.e., $c(x) = c\cdot x/d$ for $c > 0$.
Note that $c(x)$ is non-decreasing, non-negative and satisfies
$c(d) > 0$.

\begin{corollary}
  The Linear Model with $0 < c \le 1$ is Pareto-conform.
\end{corollary}
\begin{proof}
  We check the requirements of~\Cref{thm:quotient-model-pareto-conform}. First,
  $c(d) = c \le 1$.  And secondly,
  $c(x) \cdot (1 - c(x)) - x \cdot c'(x) = cx/d - c^2x^2/d^2 - cx/d
  \le 0$.
\end{proof}
We note that the Linear Model with $c \le 1$ is less conservative than
the User Equilibrium Model, i.e., more agents use the alternative path, in
particular if only few agents would use it based on its cost.
Moreover, a lower $c$ makes the model less conservative.

\section{Complexity of SAP}\label{sec:np}
Our algorithms introduced in \Cref{sec:sap-algos} have an exponential worst case running time. In the following we want to show that solving the SAP problem is indeed hard, since the corresponding decision problem, i.e., deciding whether a specific problem instance is
solvable with an overall travel time of at most \(t\), is NP-complete.
\begin{theorem}\label{thm:sap-npc}
	SAP is \textbf{NP}-complete.
\end{theorem}
\begin{proof}
	The problem is in \textbf{NP}, since a suitable alternative path $P$ can be guessed and verified in polynomial time. For showing \textbf{NP}-hardness we reduce from \textsc{SubsetSum}, which is well-known to be NP-complete~\cite{Garey1983-GARCAI-2}. Formally, in \textsc{SubsetSum} we are given a finite set $M \subset \mathbb{N}$ and a target number $w \in \mathbb{N}$ and the problem is to decide whether there exists a subset $M' \subseteq M$ such that $w=\sum_{m\in M'}m$.
	For an arbitrary \textsc{SubsetSum} instance $(M,w)$, let  $m_i$ denote the $i$-th element from $M$, for all $1\leq i \leq |M|$, and let $s = \sum_{m\in M}m$.

	Towards reducing \textsc{SubsetSum} to SAP, we construct for the instance $(M,w)$ the directed graph \(G_{M,w}=(V=\set{v_0, v_1, \ldots, v_{\abs{M}}}, E=\set{e_{i,j}\mid 1\leq i \leq |M| \wedge j\in \{1,2\}}\cup\set{e_Q})\) where, for $1\leq i \leq |M|$ and $j\in \{1,2\}$, edges $e_{i,j}$ connect vertex $v_{i-1}$ to $v_i$ and edge $e_Q$ connects $v_0$ to $v_{|M|}$.  See~\Cref{fig:subsetsum-graph} for an illustration.
	We set the cost functions as follows: \(\tau_{e_{i,1}}(x) = m_i \cdot x\), \(\tau_{e_{i,2}}(x) = m_i\) and \(\tau_{e_Q}(x)=sx + s\), for all \(x\).
	In the corresponding SAP problem on $G_{M,w}$ we search for a path from \(v_0\) to \(v_{\lvert M\vert}\), with the original path \(Q=(e_Q)\). For any \(v_0,v_\lvert M \rvert\)-path \(P\neq Q\) the psychological model \(p\) is defined as \begin{align*}
	p_{\set{P,Q}}(P)=
	\begin{cases}
	1, &\text{if }\forall x\colon \tau_P(x)=wx+s - w;\\
	0, &\text{otherwise.}\end{cases}
	\end{align*}
	We now show that this model is Pareto-conform. Let \(P_1, P_2\neq Q\) be
	\(v_0,v_{\lvert M \rvert}\)-paths. If
	\(\tau_{P_1}=\tau_{P_2}\), then \(\Tcost_{P_1}=\Tcost_{P_2}\). Otherwise,
	because \(\tau_{P_1}(1)=\tau_{P_2}(1)\) and since the functions inter- sect in
	\([0,\infty)\) at exactly one point, without loss of generality \(\tau_{P_1}(0)<\tau_{P_2}(0)\)
	and \(\tau_{P_1}(2)>\tau_{P_2}(2)\). Thus neither \(P_1\preceq P_2\) nor
	\(P_1\succeq P_2\) and the psychological model is Pareto-conform for any demand \(d\ge
	2\).

	We now show that $(M,w)$ is a yes-instance of \textsc{SubsetSum} if and only if
	for demand $d=2$ there is an alternative path to \(Q\), such that the overall
	travel time is strictly less than~\(6s\).
	Let~\(P\) be a \(v_0,v_{\lvert M \rvert}\)-path not using \(e_Q\). We define
	\(\mathbb{I}_P=\set{i \mid e_{i,1}\in P}\). Note that this set uniquely
	identifies \(P\). Then the latency of path $P$ is given by \(\tau_P=\sum_{i\in
		\mathbb{I}_P}m_i x + s -\sum_{i\in \mathbb{I}_P}m_i\).

	Assume $(M,w)$ is a yes-instance of \textsc{SubsetSum}. Let \(M'\subseteq M\)
	such that \(w=\sum_{m\in M'}m\). Let \(P\) be the path with \(\mathbb{I}_P=
	\set{ i \mid m_i \in M'}\). Then \(p_{\set{P,Q}}(P)=1\) and thus the overall
	travel time is
	\(d\cdot\tau_{P}(d)=2(2w+s-w)<6s\) when \(P\) is suggested as alternative to $Q$.

	For the other direction, if there is a path \(P\), such that the overall travel
	time is less than \(6s\) when
	\(P\) is suggested as alternative to $Q$, we know that \(p_{\set{P,Q}}(P)>0\) as
	otherwise the overall travel time would be \(6s=d\tau_Q(d)\). Thus, we must have
	\(p_{\{P,Q\}}(P) = 1\). By definition, this means that \(\sum_{i\in
		\mathbb{I}_P}m_i=w\) and thus there is a subset \(M'\subseteq M\), such that
	\(w=\sum_{m\in M'}m\) and hence $(M,w)$ is a yes-instance of \textsc{SubsetSum}.
	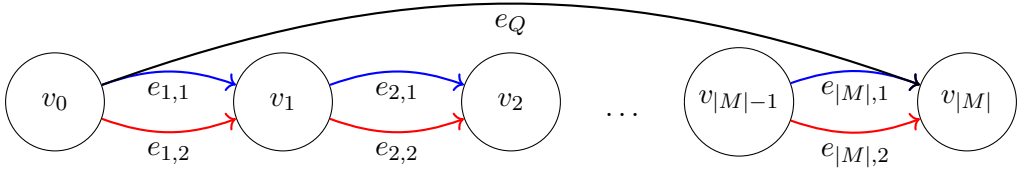
\begin{figure}
		\begin{center}
			\begin{tikzpicture}[auto, thin,
			arc1l/.style={bend left=20},
			arc2l/.style={bend left=20},
			arc1r/.style={bend right=20}]
			\begin{scope}[every node/.style={circle,draw=black,minimum size=1.3cm}]
			\node (0) at (-3,0) {\(v_0\)};
			\node (1) at (0,0) {\(v_1\)};
			\node (2) at (3,0) {\(v_2\)};
			\node (3) at (6,0) {\(v_{\abs{M}-1}\)};
			\node (4) at (9,0) {\(v_{\abs{M}}\)};
			\end{scope}
			\begin{scope}[every edge/.style={draw=blue,thick}]
			\draw[->, arc2l]  (0) edge node[below]{\(e_{1,1}\)} (1);
			\draw[->, arc2l]  (1) edge node[below]{\(e_{2,1}\)} (2);
			\draw[->, arc2l]  (3) edge node[below]{\(e_{\abs{M},1}\)} (4);
			\end{scope}
			\begin{scope}[every edge/.style={draw=red,thick}]
			\draw[->, arc1r]  (0) edge node[below]{\(e_{1,2}\)} (1);
			\draw[->, arc1r]  (1) edge node[below]{\(e_{2,2}\)} (2);
			\draw[->, arc1r]  (3) edge node[below]{\(e_{\abs{M},2}\)} (4);
			\end{scope}
			\begin{scope}[every edge/.style={draw=black,thick}]
			\draw[->, arc1l]  (0) edge node[below]{\(e_Q\)} (4);
			\end{scope}
			\begin{scope}[every edge/.style={draw=white,thick}]
			\draw[->]  (2) edge node[below]{\(\ldots\)} (3);
			\end{scope}
			\end{tikzpicture}
		\end{center}
		\caption{Constructed graph \(G_{M,w}\) for the reduction from \textsc{SubsetSum}.}\label{fig:subsetsum-graph}
	\end{figure}
\end{proof}

It directly follows that our proposed variants D-SAP and $1$D-SAP are also NP-complete.
\begin{corollary}
	D-SAP and 1D-SAP are \textbf{NP}-complete.
\end{corollary}
\begin{proof}
	Every valid solution for SAP needs to be completely disjoint in the proof above. Thus the theorem holds for the case of restricted
	disjointedness as well.
\end{proof}


\section{Empirical Evaluation}\label{sec:evaluation}

\begin{figure}[b]
  \centering
  \includegraphics[width=0.46\linewidth]{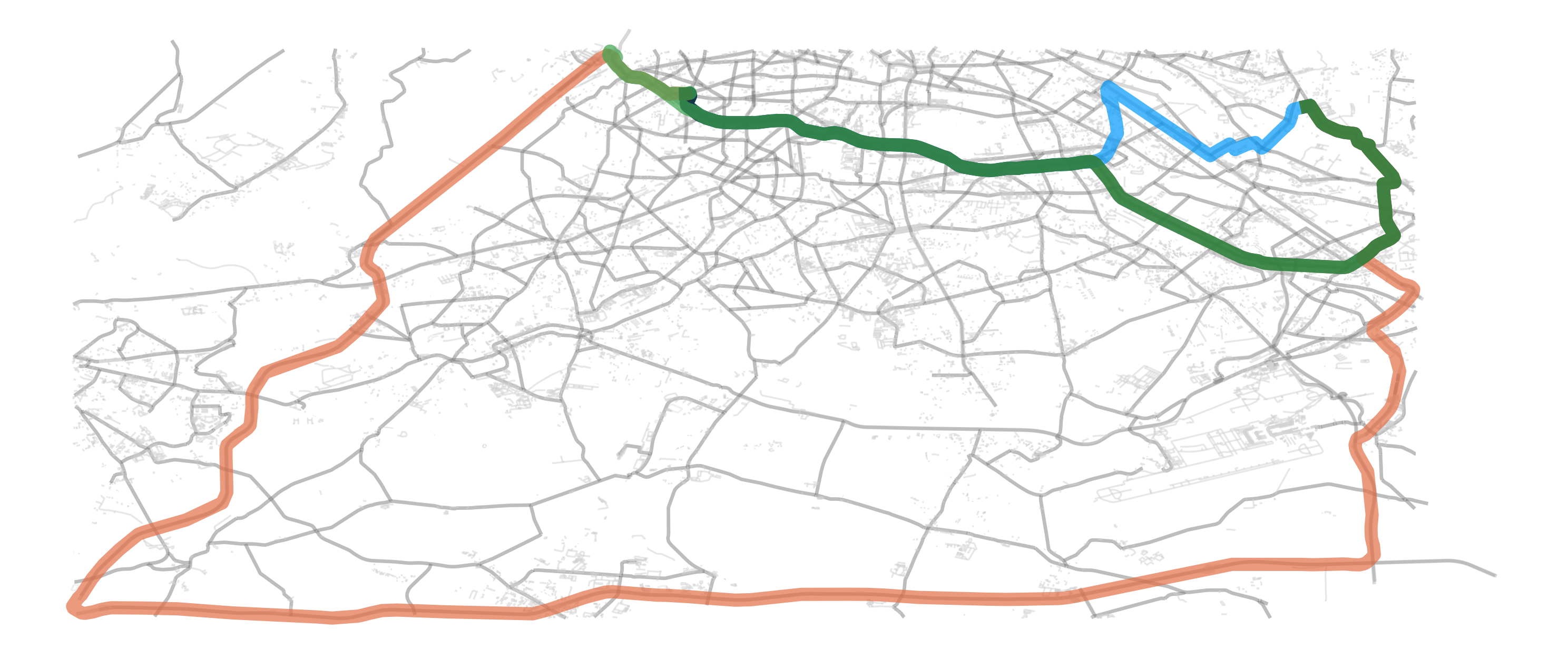}
  \caption{Visualization of example routes: original route (blue),
    optimal alternative routes with respect to the User Equilibrium
    for SAP (green), for 1D-SAP (black), and D-SAP (orange).}
  \label{fig:example-routes}
\end{figure}

In this section we fix implementation details and evaluate the
proposed algorithms.  Our evaluation focuses on the following aspects.
\begin{description}
\item[Performance.] Are the algorithms sufficiently efficient for
  practical problem instances?  How do the different algorithms
  compare in terms of run time?
\item[Strategic Improvement.] How much does strategic routing improve
  the overall travel time?  How does the requirement of disjoint or
  $1$-disjoint alternatives impact this improvement?
\end{description}
Additional evaluation regarding the psychological models can be found
in \Cref{subpar:psychmods}. 
For now, we fix the psychological model to be the User
Equilibrium.  


%
We model cost functions \(\tau_e\) as proposed by the U.S.\ Bureau of
Public Roads~\cite{us1964}, i.e., for parameters
\(\alpha, \beta \geq 0\), we have
$\tau_e(x) = \ell_e/s_e \cdot (1+\alpha (x/c_e)^\beta)$ where \(s_e\),
\(c_e\), and \(\ell_e\) denote free flow speed, capacity and length of
$e$.  We set \(\alpha = 0.15\) and \(\beta = 2\).  Thus, for
appropriate $a$ and $b$, we get canonical cost functions of the form
\(\tau_e(x) = ax^2 +b\) as defined in \Cref{sec:sap-algos}.

For solving the multi-criteria shortest path problem, we implement a
multi-criteria A* variant~\cite{Mandow_DeLaCruz_2008}.  As lower
bound, we use the distances in the parameters \(a\) and \(b\) to
\(t\). These distances are calculated using two runs of Dijkstra's
algorithm. We note that A* solves a multi-target shortest path
problem, which we need for two algorithms; see \cref{sec:1d-sap-dp}.
For calculating the Pareto-frontiers we use the simple cull
algorithm~\cite{Yukish_2004}.

We use the following naming scheme.  We abbreviate the algorithms
from~\Cref{sec:single-altern,subsec:1d_SAP} with SAP and 1D-SAP,
respectively.  We denote the \emph{fewer criteria} (FC) approaches
with D-SAP (\Cref{sec:disjoint-single-alt-path}) 1D-SAP-FC
(\cref{sec:1d-sap-dp}) and SAP-FC (\cref{sec:dynamic-program-sap}).
To evaluate the strategic improvement, we compare them to the solution
of proposing only the shortest path to all agents, assuming either one
single agent (\(1\)-SP) or \(d\) agents (\(d\)-SP) on every edge.

We test our implementations on the street network of Berlin, Germany
with 75 origin--destination pairs (OD-pairs), randomly chosen from
real-world OD-pairs. The OD-pairs as well as the network were provided by TomTom. For every OD-pair, we set $Q$
to the shortest route for a single agent and run all algorithms
for our psychological models and demands
\(d \in \set{100, 500, 1000, 1500, 2000, 2500, 3000}\).  One unit of
demand represents $7$--$20$ vehicles per hour.  The
imprecision is due to the fact that the exact penetration rate of
TomTom devices is unknown and that the map data is given with
respect to only TomTom users.

All experiments have been conducted on a machine with two Intel Xeon
Gold 5118 (12-core) CPUs with 64GiB of memory. The
multi-criteria shortest-path calculations of SAP-FC and 1D-SAP-FC
have been parallelized to 20 threads. 




\subparagraph*{Run Time.}


\begin{figure}[t]
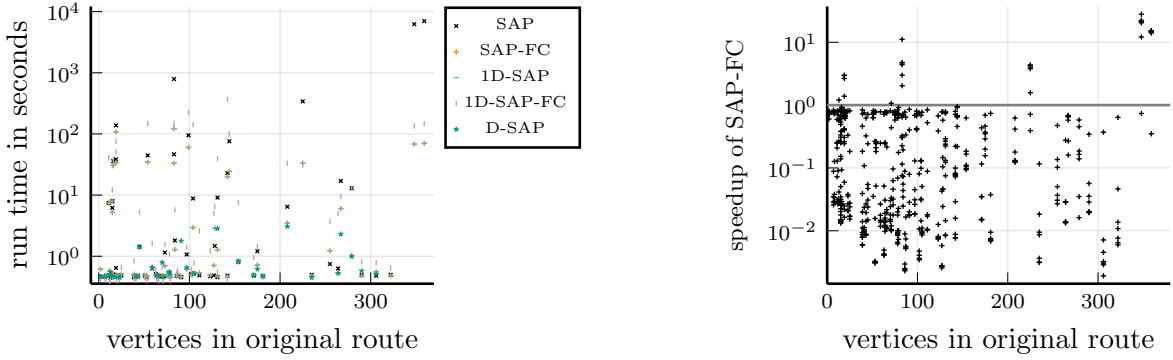

  \centering
  \input{content/plots/runtime_by_length_d2000}\hfill
  \input{content/plots/nn_dp_runtime_relation}
  \caption{Left: Absolute run times.  Each point represents one
    OD-pair for demand $d = 2000$.  For D-SAP, we excluded the
    OD-pairs that did not have a solution that was disjoint from the
    original route.  Right: Speedup of SAP-FC over SAP, with one point
    for each OD-pair and each value of $d$.}
  \label{fig:dp_comparison}
\end{figure}
\tikzexternaldisable

\Cref{fig:dp_comparison} shows the run times of our algorithms,
depending on the length of the original route.  
The main takeaways from \Cref{fig:dp_comparison} (left) are that
requiring disjoint routes makes the problem easier
and that the algorithms requiring fewer criteria but more
multi-criteria shortest path queries are faster for instances with
long original routes.  \Cref{fig:dp_comparison} (right) shows the
speedup of SAP-FC over SAP.  One can see that SAP is actually faster
than SAP-FC for most instances, sometimes up to two orders of
magnitude.  However, these are the instances with short original
route, which have low run times anyways.  On the other hand, SAP-FC is
up to one order of magnitude faster than SAP on some instances with
long original path.  We note that the multi-criteria shortest path
queries in SAP-FC can be parallelized, and we used 20 threads in our
experiments.  However, this parallelization cannot explain such high
speedups.  In \Cref{fig:dp_comparison} (left), one can see that SAP-FC
actually has rather consistent run times compared to SAP and never
exceeded 30 minutes.
Thus, our observations show that we can feasibly solve the problems
SAP and even more so 1D-SAP in the context of small distance queries,
e.g., in city networks, despite the worst-case exponential running
time.

\subparagraph*{Strategic Improvement.}

We assess how much strategic routing gains in terms of travel time
with respect to different disjointedness.  \Cref{fig:example-routes}
shows solutions for SAP, 1D-SAP and D-SAP routes.
%
The resulting travel times are shown in
\cref{fig:otScores}.  We see that the larger the number of agents, the
more we benefit from strategic routing.
\tikzexternalenable
\begin{figure}[t]
  \centering
  \begin{tikzpicture}[/tikz/background rectangle/.style={fill={rgb,1:red,1.0;green,1.0;blue,1.0}, draw opacity={1.0}}, show background rectangle]
\begin{axis}[title={}, title style={at={{(0.5,1)}}, font={{\fontsize{14 pt}{18.2
    pt}\selectfont}}, color={rgb,1:red,0.0;green,0.0;blue,0.0}, draw
opacity={1.0}, rotate={0.0}}, legend
style={color={rgb,1:red,0.0;green,0.0;blue,0.0}, draw opacity={1.0}, line
width={1}, solid, fill={rgb,1:red,1.0;green,1.0;blue,1.0}, fill opacity={1.0},
text opacity={1.0}, font={{\fontsize{8 pt}{10.4 pt}\selectfont}}}, axis
background/.style={fill={rgb,1:red,1.0;green,1.0;blue,1.0}, opacity={1.0}},
anchor={north west}, xshift={0.0mm}, yshift={-0.0mm}, height=11\baselineskip,
xlabel near ticks, ylabel near ticks, legend pos={north west}, scaled x
ticks={false}, xlabel={$d$}, x tick
style={color={rgb,1:red,0.0;green,0.0;blue,0.0}, opacity={1.0}}, x tick label
style={color={rgb,1:red,0.0;green,0.0;blue,0.0}, opacity={1.0}, rotate={0}},
xlabel style={, font={{\fontsize{11 pt}{14.3 pt}\selectfont}},
color={rgb,1:red,0.0;green,0.0;blue,0.0}, draw opacity={1.0}, rotate={0.0}},
xmajorgrids={true}, xmin={13.0}, xmax={3087.0},
xtick={{500.0,1000.0,1500.0,2000.0,2500.0,3000.0}},
xticklabels={{$500$,$1000$,$1500$,$2000$,$2500$,$3000$}}, xtick align={inside},
xticklabel style={font={{\fontsize{8 pt}{10.4 pt}\selectfont}},
color={rgb,1:red,0.0;green,0.0;blue,0.0}, draw opacity={1.0}, rotate={0.0}}, x
grid style={color={rgb,1:red,0.0;green,0.0;blue,0.0}, draw opacity={0.1}, line
width={0.5}, solid}, axis x line*={left}, x axis line
style={color={rgb,1:red,0.0;green,0.0;blue,0.0}, draw opacity={1.0}, line
width={1}, solid}, scaled y ticks={false}, ylabel={travel time/agent in s}, y
tick style={color={rgb,1:red,0.0;green,0.0;blue,0.0}, opacity={1.0}}, y tick
label style={color={rgb,1:red,0.0;green,0.0;blue,0.0}, opacity={1.0},
rotate={0}}, ylabel style={, font={{\fontsize{8.5 pt}{14.3 pt}\selectfont}}, color={rgb,1:red,0.0;green,0.0;blue,0.0}, draw opacity={1.0}, rotate={0.0}}, ymajorgrids={true}, ymin={-2426.2627228828833}, ymax={111784.68898114115}, ytick={{0.0,25000.0,50000.0,75000.0,100000.0}}, yticklabels={{$0$,$2.5\times10^{4}$,$5.0\times10^{4}$,$7.5\times10^{4}$,$1.0\times10^{5}$}}, ytick align={inside}, yticklabel style={font={{\fontsize{8 pt}{10.4 pt}\selectfont}}, color={rgb,1:red,0.0;green,0.0;blue,0.0}, draw opacity={1.0}, rotate={0.0}}, y grid style={color={rgb,1:red,0.0;green,0.0;blue,0.0}, draw opacity={0.1}, line width={0.5}, solid}, axis y line*={left}, y axis line style={color={rgb,1:red,0.0;green,0.0;blue,0.0}, draw opacity={1.0}, line width={1}, solid}, colorbar style={title={}, point meta max={nan}, point meta min={nan}}]
    \addplot[color=brightlavender, name path={28f1e422-3528-494a-b620-441d64370b2d}, draw opacity={1.0}, line width={1}, solid]
        coordinates {
            (100.0,878.8960333333337)
            (500.0,3753.3817333333336)
            (1000.0,12736.143866666667)
            (1500.0,27707.421066666666)
            (2000.0,48667.21533333333)
            (2500.0,75615.50453333334)
            (3000.0,108552.30355555557)
        }
        ;
    \addlegendentry {$1$-SP}
    \addplot[color=orange, name path={83030a37-eea9-4492-9ece-47067b6dca57}, draw opacity={1.0}, line width={1}, solid]
        coordinates {
            (100.0,862.3952733333334)
            (500.0,2776.7481066666664)
            (1000.0,8516.972)
            (1500.0,17993.124533333335)
            (2000.0,31240.899266666667)
            (2500.0,48266.2384)
            (3000.0,69067.47555555555)
        }
        ;
    \addlegendentry {$d$-SP}
    \addplot[color={rgb,1:red,0.2422;green,0.6433;blue,0.3044}, name path={351f996c-119a-4f27-9c2a-ef48633c484a}, draw opacity={1.0}, line width={1}, solid]
        coordinates {
            (100.0,875.3271200000001)
            (500.0,2823.2123733333337)
            (1000.0,8716.448133333333)
            (1500.0,18518.213066666667)
            (2000.0,32228.64633333333)
            (2500.0,49850.1872)
            (3000.0,71387.264)
        }
        ;
    \addlegendentry {D-SAP}
    \addplot[color=cyan, name path={5e5fb48f-4cfd-4c66-a668-a5b114ebebf3}, draw opacity={1.0}, line width={1}, solid]
        coordinates {
            (100.0,806.1227027027027)
            (500.0,1836.6606756756757)
            (1000.0,4841.925108108108)
            (1500.0,9787.042702702702)
            (2000.0,16682.95047297297)
            (2500.0,25522.680972972976)
            (3000.0,36315.517567567564)
        }
        ;
    \addlegendentry {SAP}
    \addplot[color={rgb,1:red,0.6755;green,0.5557;blue,0.0942}, name path={884b86e0-e762-4bb1-a42b-c643adb38ee8}, draw opacity={1.0}, line width={1}, solid]
        coordinates {
            (100.0,844.8603999999999)
            (500.0,1956.74928)
            (1000.0,5140.114373333333)
            (1500.0,10355.042666666666)
            (2000.0,17613.557133333332)
            (2500.0,26915.58176)
            (3000.0,38274.36488888889)
        }
        ;
    \addlegendentry {1D-SAP}
\end{axis}
\end{tikzpicture}\hfill
  \begin{tikzpicture}[/tikz/background rectangle/.style={fill={rgb,1:red,1.0;green,1.0;blue,1.0}, draw opacity={1.0}}, show background rectangle]
\begin{axis}[title={}, title style={at={{(0.5,1)}}, font={{\fontsize{14 pt}{18.2
    pt}\selectfont}}, color={rgb,1:red,0.0;green,0.0;blue,0.0}, draw
opacity={1.0}, rotate={0.0}}, legend
style={color={rgb,1:red,0.0;green,0.0;blue,0.0}, draw opacity={1.0}, line
width={1}, solid, fill={rgb,1:red,1.0;green,1.0;blue,1.0}, fill opacity={1.0},
text opacity={1.0}, font={{\fontsize{8 pt}{10.4 pt}\selectfont}}}, axis
background/.style={fill={rgb,1:red,1.0;green,1.0;blue,1.0}, opacity={1.0}},
anchor={north west}, xshift={0.0mm}, yshift={-0.0mm}, height=11\baselineskip,
xlabel near ticks, ylabel near ticks, legend pos={outer north east}, scaled x
ticks={false}, xlabel={$d$}, x tick
style={color={rgb,1:red,0.0;green,0.0;blue,0.0}, opacity={1.0}}, x tick label
style={color={rgb,1:red,0.0;green,0.0;blue,0.0}, opacity={1.0}, rotate={0}},
xlabel style={, font={{\fontsize{11 pt}{14.3 pt}\selectfont}},
color={rgb,1:red,0.0;green,0.0;blue,0.0}, draw opacity={1.0}, rotate={0.0}},
xmajorgrids={true}, xmin={13.0}, xmax={3087.0},
xtick={{500.0,1000.0,1500.0,2000.0,2500.0,3000.0}},
xticklabels={{$500$,$1000$,$1500$,$2000$,$2500$,$3000$}}, xtick align={inside},
xticklabel style={font={{\fontsize{8 pt}{10.4 pt}\selectfont}},
color={rgb,1:red,0.0;green,0.0;blue,0.0}, draw opacity={1.0}, rotate={0.0}}, x
grid style={color={rgb,1:red,0.0;green,0.0;blue,0.0}, draw opacity={0.1}, line
width={0.5}, solid}, axis x line*={left}, x axis line
style={color={rgb,1:red,0.0;green,0.0;blue,0.0}, draw opacity={1.0}, line
width={1}, solid}, scaled y ticks={false}, ylabel={travel time relative to
$d$-SP}, y tick style={color={rgb,1:red,0.0;green,0.0;blue,0.0}, opacity={1.0}},
y tick label style={color={rgb,1:red,0.0;green,0.0;blue,0.0}, opacity={1.0},
rotate={0}}, ylabel style={, font={{\fontsize{8.5 pt}{14.3 pt}\selectfont}}, color={rgb,1:red,0.0;green,0.0;blue,0.0}, draw opacity={1.0}, rotate={0.0}}, ymajorgrids={true}, ymin={0.5070159544586358}, ymax={1.6026946027535611}, ytick={{0.75,1.0,1.25,1.5}}, yticklabels={{$0.75$,$1.00$,$1.25$,$1.50$}}, ytick align={inside}, yticklabel style={font={{\fontsize{8 pt}{10.4 pt}\selectfont}}, color={rgb,1:red,0.0;green,0.0;blue,0.0}, draw opacity={1.0}, rotate={0.0}}, y grid style={color={rgb,1:red,0.0;green,0.0;blue,0.0}, draw opacity={0.1}, line width={0.5}, solid}, axis y line*={left}, y axis line style={color={rgb,1:red,0.0;green,0.0;blue,0.0}, draw opacity={1.0}, line width={1}, solid}, colorbar style={title={}, point meta max={nan}, point meta min={nan}}]
    \addplot[color=brightlavender, name path={aa54f112-0423-4b75-81f5-3b423a7df513}, draw opacity={1.0}, line width={1}, solid]
        coordinates {
            (100.0,1.019133639191019)
            (500.0,1.3517184811693497)
            (1000.0,1.4953840245883945)
            (1500.0,1.5398893624804861)
            (2000.0,1.557804559910993)
            (2500.0,1.5666334696870294)
            (3000.0,1.5716848296886103)
        }
        ;
    \addplot[color=orange, name path={aa30ec1a-36d0-4b1e-8ddb-7073620b8a48}, draw opacity={1.0}, line width={1}, solid]
        coordinates {
            (100.0,1.0)
            (500.0,1.0)
            (1000.0,1.0)
            (1500.0,1.0)
            (2000.0,1.0)
            (2500.0,1.0)
            (3000.0,1.0)
        }
        ;
    \addplot[color={rgb,1:red,0.2422;green,0.6433;blue,0.3044}, name path={d35f3cc8-19a4-4fd1-bc2f-e74fe0309a4c}, draw opacity={1.0}, line width={1}, solid]
        coordinates {
            (100.0,1.0149952661691692)
            (500.0,1.016733338740778)
            (1000.0,1.023421015512712)
            (1500.0,1.0291827321241829)
            (2000.0,1.031617113778814)
            (2500.0,1.0328169099666156)
            (3000.0,1.0335872771631633)
        }
        ;
    \addplot[color=cyan, name path={b5a3764b-81f4-450e-b54a-5d37cfa127a4}, draw opacity={1.0}, line width={1}, solid]
        coordinates {
            (100.0,0.9741171653517537)
            (500.0,0.6809253201077545)
            (1000.0,0.5828885369565554)
            (1500.0,0.5571150100085)
            (2000.0,0.5466818902106924)
            (2500.0,0.5411748769071669)
            (3000.0,0.5380257275235865)
        }
        ;
    \addplot[color={rgb,1:red,0.6755;green,0.5557;blue,0.0942}, name path={8aaa8809-85b3-435d-8cf3-5be4e81b5412}, draw opacity={1.0}, line width={1}, solid]
        coordinates {
            (100.0,0.9796672432287834)
            (500.0,0.7046909567713616)
            (1000.0,0.6035142974913306)
            (1500.0,0.5754999720856337)
            (2000.0,0.5637980194803995)
            (2500.0,0.55764821648086)
            (3000.0,0.5541590246497758)
        }
        ;
\end{axis}
\end{tikzpicture}
  \caption{The plots show the travel time per agent depending on the
    demand $d$, where each data point is averaged over all
    OD-pairs.  Absolute values are shown on the left, relative values
    with respect to the $d$-SP solution are shown on the right.}
  \label{fig:otScores}
\end{figure}
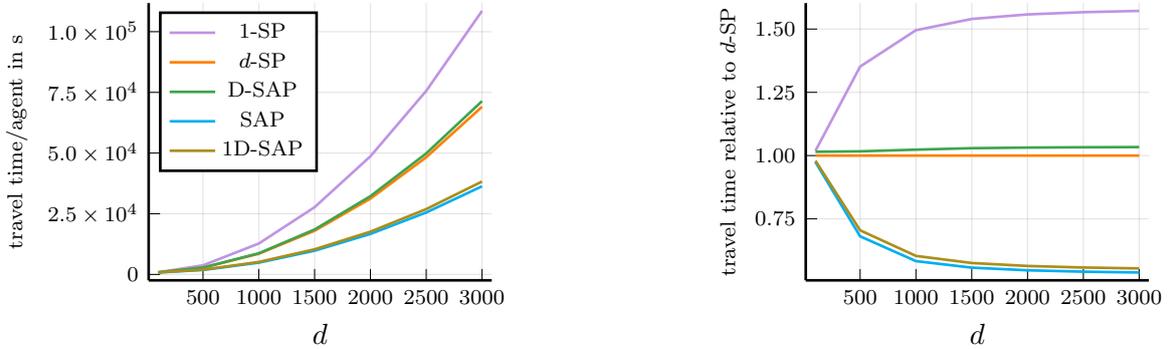
The plots show that, in direct comparison to the shortest path
assuming \(d\) agents per edge (\(d\)-SP), the SAP algorithms yield
results of about \(50\,\%\) reduced travel time for growing values of
\(d\).  Constraining the alternative route to be 1-disjoint from the
original only has a slight disadvantage (on average 1D-SAP is worse by
$2.2\,\%$).  Thus, taking into account that 1D-SAP can be solved
faster, solving 1D-SAP might give a good trade-off between run time
and quality of the solution.
Demanding full disjointedness leads to much worse travel times, as in
62.3\,\% of our test cases, no fully disjoint alternative exists, due
to the graph structure.  In this case, we assume that all agents use
the original route.  Restricted to the instances that allow for a
fully disjoint solution, the solution to D-SAP on average leads to a
$11.4\,\%$ higher travel time per agent compared to 1D-SAP.

\subparagraph*{Psychological Models.}\label{subpar:psychmods}

\begin{figure}[t]
  \centering
  \includegraphics[width=0.5\linewidth]{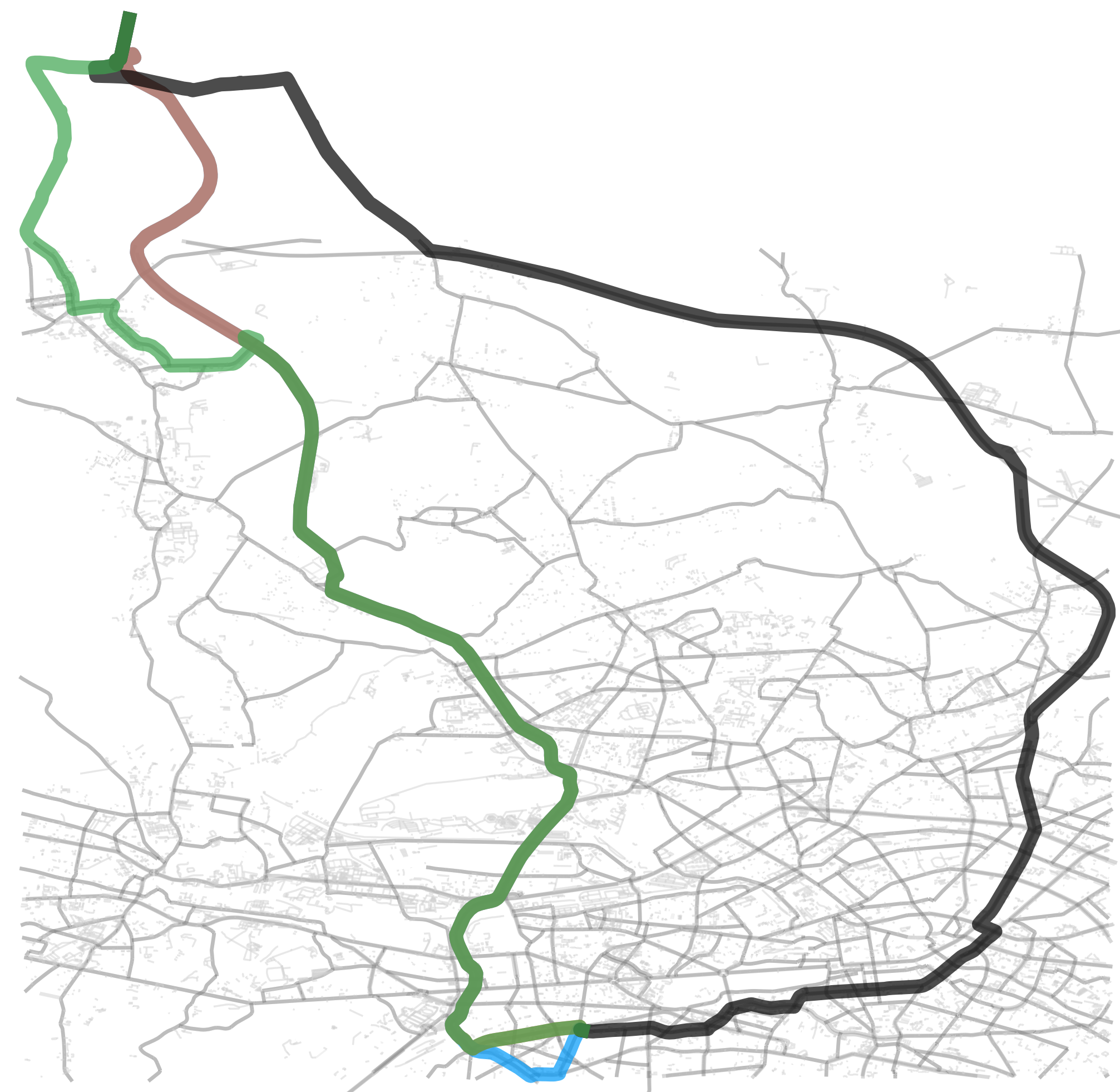}
  \caption{Visualization of example routes:  Original route (blue),
    and optimal alternative routes for the SAP problem with respect to
    System Optimum (black), User Equilibrium (green) and Linear Model
    (orange).}\label{fig:example-routes-psychmod}
\end{figure}

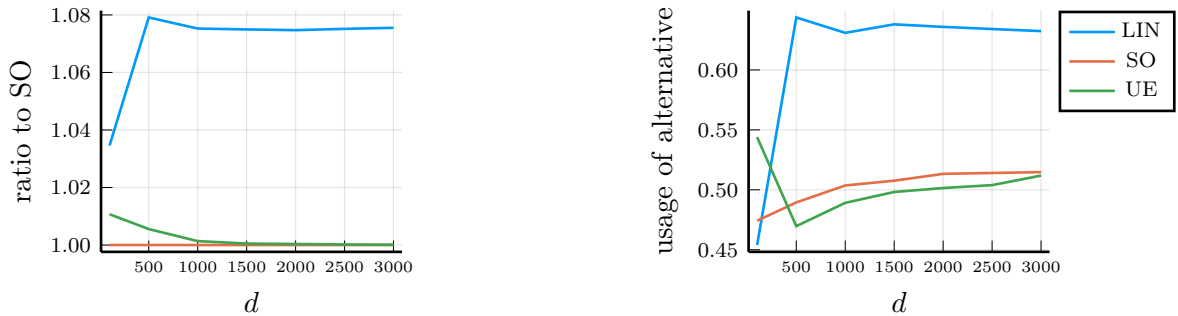
\begin{figure}[t]
  \centering
  \begin{tikzpicture}[/tikz/background rectangle/.style={fill={rgb,1:red,1.0;green,1.0;blue,1.0}, draw opacity={1.0}}, show background rectangle]
\begin{axis}[title={}, title style={at={{(0.5,1)}}, font={{\fontsize{14 pt}{18.2
    pt}\selectfont}}, color={rgb,1:red,0.0;green,0.0;blue,0.0}, draw
opacity={1.0}, rotate={0.0}}, legend
style={color={rgb,1:red,0.0;green,0.0;blue,0.0}, draw opacity={1.0}, line
width={1}, solid, fill={rgb,1:red,1.0;green,1.0;blue,1.0}, fill opacity={1.0},
text opacity={1.0}, font={{\fontsize{8 pt}{10.4 pt}\selectfont}}}, axis
background/.style={fill={rgb,1:red,1.0;green,1.0;blue,1.0}, opacity={1.0}},
anchor={north west}, xshift={0.0mm}, yshift={-0.0mm}, height=10\baselineskip,
xlabel near ticks, ylabel near ticks, legend pos={outer north east}, scaled x
ticks={false}, xlabel={$d$}, x tick
style={color={rgb,1:red,0.0;green,0.0;blue,0.0}, opacity={1.0}}, x tick label
style={color={rgb,1:red,0.0;green,0.0;blue,0.0}, opacity={1.0}, rotate={0}},
xlabel style={, font={{\fontsize{11 pt}{14.3 pt}\selectfont}},
color={rgb,1:red,0.0;green,0.0;blue,0.0}, draw opacity={1.0}, rotate={0.0}},
xmajorgrids={true}, xmin={13.0}, xmax={3087.0},
xtick={{500.0,1000.0,1500.0,2000.0,2500.0,3000.0}},
xticklabels={{$500$,$1000$,$1500$,$2000$,$2500$,$3000$}}, xtick align={inside},
xticklabel style={font={{\fontsize{6 pt}{10.4 pt}\selectfont}},
color={rgb,1:red,0.0;green,0.0;blue,0.0}, draw opacity={1.0}, rotate={0.0}}, x
grid style={color={rgb,1:red,0.0;green,0.0;blue,0.0}, draw opacity={0.1}, line
width={0.5}, solid}, axis x line*={left}, x axis line
style={color={rgb,1:red,0.0;green,0.0;blue,0.0}, draw opacity={1.0}, line
width={1}, solid}, scaled y ticks={false}, ylabel={ratio to SO}, y tick style={color={rgb,1:red,0.0;green,0.0;blue,0.0}, opacity={1.0}}, y tick label style={color={rgb,1:red,0.0;green,0.0;blue,0.0}, opacity={1.0}, rotate={0}}, ylabel style={, font={{\fontsize{11 pt}{14.3 pt}\selectfont}}, color={rgb,1:red,0.0;green,0.0;blue,0.0}, draw opacity={1.0}, rotate={0.0}}, ymajorgrids={true}, ymin={0.9976263028115808}, ymax={1.0814969368023903}, ytick={{1.0,1.02,1.04,1.06,1.08}}, yticklabels={{$1.00$,$1.02$,$1.04$,$1.06$,$1.08$}}, ytick align={inside}, yticklabel style={font={{\fontsize{8 pt}{10.4 pt}\selectfont}}, color={rgb,1:red,0.0;green,0.0;blue,0.0}, draw opacity={1.0}, rotate={0.0}}, y grid style={color={rgb,1:red,0.0;green,0.0;blue,0.0}, draw opacity={0.1}, line width={0.5}, solid}, axis y line*={left}, y axis line style={color={rgb,1:red,0.0;green,0.0;blue,0.0}, draw opacity={1.0}, line width={1}, solid}, colorbar style={title={}, point meta max={nan}, point meta min={nan}}]
    \addplot[color={rgb,1:red,0.0;green,0.6056;blue,0.9787}, name path={6cd61531-376b-4e8a-aa55-8f6c937abe62}, draw opacity={1.0}, line width={1}, solid]
        coordinates {
            (100.0,1.034628339111745)
            (500.0,1.0791232396139712)
            (1000.0,1.0752675780646768)
            (1500.0,1.0749623317909929)
            (2000.0,1.074708324413654)
            (2500.0,1.0751547142502842)
            (3000.0,1.0755224379975439)
        }
        ;
    \addplot[color={rgb,1:red,0.8889;green,0.4356;blue,0.2781}, name path={5bb99821-d3ea-4b31-b934-1a69c722cdd3}, draw opacity={1.0}, line width={1}, solid]
        coordinates {
            (100.0,1.0)
            (500.0,1.0)
            (1000.0,1.0)
            (1500.0,1.0)
            (2000.0,1.0)
            (2500.0,1.0)
            (3000.0,1.0)
        }
        ;
    \addplot[color={rgb,1:red,0.2422;green,0.6433;blue,0.3044}, name path={02c1b24e-ab20-466c-b3b0-6678a4c56a9b}, draw opacity={1.0}, line width={1}, solid]
        coordinates {
            (100.0,1.0106993472988464)
            (500.0,1.0055528362037722)
            (1000.0,1.0013634895144872)
            (1500.0,1.0005167476273686)
            (2000.0,1.0003590315936612)
            (2500.0,1.0001968874534604)
            (3000.0,1.0001012176948074)
        }
        ;
\end{axis}
\end{tikzpicture}\hfill
  \begin{tikzpicture}[/tikz/background rectangle/.style={fill={rgb,1:red,1.0;green,1.0;blue,1.0}, draw opacity={1.0}}, show background rectangle]
\begin{axis}[title={}, title style={at={{(0.5,1)}}, font={{\fontsize{14 pt}{18.2
    pt}\selectfont}}, color={rgb,1:red,0.0;green,0.0;blue,0.0}, draw
opacity={1.0}, rotate={0.0}}, legend
style={color={rgb,1:red,0.0;green,0.0;blue,0.0}, draw opacity={1.0}, line
width={1}, solid, fill={rgb,1:red,1.0;green,1.0;blue,1.0}, fill opacity={1.0},
text opacity={1.0}, font={{\fontsize{8 pt}{10.4 pt}\selectfont}}}, axis
background/.style={fill={rgb,1:red,1.0;green,1.0;blue,1.0}, opacity={1.0}},
anchor={north west}, xshift={0.0mm}, yshift={-0.0mm}, height=10\baselineskip,
xlabel near ticks, ylabel near ticks, legend pos={outer north east}, scaled x ticks={false}, xlabel={$d$}, x tick style={color={rgb,1:red,0.0;green,0.0;blue,0.0}, opacity={1.0}}, x tick label style={color={rgb,1:red,0.0;green,0.0;blue,0.0}, opacity={1.0}, rotate={0}}, xlabel style={, font={{\fontsize{11 pt}{14.3 pt}\selectfont}}, color={rgb,1:red,0.0;green,0.0;blue,0.0}, draw opacity={1.0}, rotate={0.0}}, xmajorgrids={true}, xmin={13.0}, xmax={3087.0}, xtick={{500.0,1000.0,1500.0,2000.0,2500.0,3000.0}}, xticklabels={{$500$,$1000$,$1500$,$2000$,$2500$,$3000$}}, xtick align={inside}, xticklabel style={font={{\fontsize{6 pt}{10.4 pt}\selectfont}}, color={rgb,1:red,0.0;green,0.0;blue,0.0}, draw opacity={1.0}, rotate={0.0}}, x grid style={color={rgb,1:red,0.0;green,0.0;blue,0.0}, draw opacity={0.1}, line width={0.5}, solid}, axis x line*={left}, x axis line style={color={rgb,1:red,0.0;green,0.0;blue,0.0}, draw opacity={1.0}, line width={1}, solid}, scaled y ticks={false}, ylabel={usage of alternative}, y tick style={color={rgb,1:red,0.0;green,0.0;blue,0.0}, opacity={1.0}}, y tick label style={color={rgb,1:red,0.0;green,0.0;blue,0.0}, opacity={1.0}, rotate={0}}, ylabel style={, font={{\fontsize{11 pt}{14.3 pt}\selectfont}}, color={rgb,1:red,0.0;green,0.0;blue,0.0}, draw opacity={1.0}, rotate={0.0}}, ymajorgrids={true}, ymin={0.4482229729729729}, ymax={0.6494797297297297}, ytick={{0.45,0.5,0.55,0.6}}, yticklabels={{$0.45$,$0.50$,$0.55$,$0.60$}}, ytick align={inside}, yticklabel style={font={{\fontsize{8 pt}{10.4 pt}\selectfont}}, color={rgb,1:red,0.0;green,0.0;blue,0.0}, draw opacity={1.0}, rotate={0.0}}, y grid style={color={rgb,1:red,0.0;green,0.0;blue,0.0}, draw opacity={0.1}, line width={0.5}, solid}, axis y line*={left}, y axis line style={color={rgb,1:red,0.0;green,0.0;blue,0.0}, draw opacity={1.0}, line width={1}, solid}, colorbar style={title={}, point meta max={nan}, point meta min={nan}}]
    \addplot[color={rgb,1:red,0.0;green,0.6056;blue,0.9787}, name path={d09eca64-b8ea-4d94-b73d-eb0ccb0e1de0}, draw opacity={1.0}, line width={1}, solid]
        coordinates {
            (100.0,0.4539189189189189)
            (500.0,0.6437837837837838)
            (1000.0,0.6309589041095891)
            (1500.0,0.6380469483568076)
            (2000.0,0.6359225352112675)
            (2500.0,0.6341464788732394)
            (3000.0,0.6324037558685446)
        }
        ;
    \addlegendentry {LIN}
    \addplot[color={rgb,1:red,0.8889;green,0.4356;blue,0.2781}, name path={6a4bccfb-fbc8-4e2c-8871-645997ea89be}, draw opacity={1.0}, line width={1}, solid]
        coordinates {
            (100.0,0.4741891891891891)
            (500.0,0.4894722222222222)
            (1000.0,0.5036164383561642)
            (1500.0,0.5076073059360731)
            (2000.0,0.5132986111111109)
            (2500.0,0.514038888888889)
            (3000.0,0.5147546296296296)
        }
        ;
    \addlegendentry {SO}
    \addplot[color={rgb,1:red,0.2422;green,0.6433;blue,0.3044}, name path={876c5c3f-8b79-4179-bdfe-c13c37fe4f79}, draw opacity={1.0}, line width={1}, solid]
        coordinates {
            (100.0,0.543918918918919)
            (500.0,0.4696712328767124)
            (1000.0,0.48912162162162154)
            (1500.0,0.49817117117117093)
            (2000.0,0.5014797297297299)
            (2500.0,0.5039027027027029)
            (3000.0,0.5118564814814817)
        }
        ;
    \addlegendentry {UE}
\end{axis}
\end{tikzpicture}
  \caption{Left: Overall travel time of the User Equilibrium and the Linear
    Model, relative to the System Optimum, depending on the demand
    $d$.  Each data point is averaged over all OD-pairs.  Right: The
    fraction of flow $x_P/d$ on the alternative route depending on the
    $d$ and the psychological model; again averaged over all
    OD-pairs.}\label{fig:ue-so-usage}
\end{figure}

The gain of splitting traffic also depends on the psychological model.
We note that for most OD-pairs, the different psychological models
lead to the same route and only differ in the amount of traffic using
the alternative.  However, there are examples where we actually get
different routes, see, e.g., \Cref{fig:example-routes-psychmod}.  In
this particular instance it is interesting to see that the System
Optimum would require a large detour for some drivers (around
$17\,\%$), which they probably would not accept without additional
incentive.

To compare the models we proposed, we examine the proportion of agents
using the suggested alternative route, and the resulting costs
compared to the overall travel time when using the System Optimum
in~\Cref{fig:ue-so-usage}. The Linear Model is configured with
\(c = 1\). One can see in~\Cref{fig:ue-so-usage} (left) that for
increasing \(d\), the User Equilibrium overall travel time approaches
that of the System Optimum, indicating that the optimal alternative
route of the User Equilibrium is not much worse than the System
Optimum. This is supported by the similar amount of agents that are
assigned by the models to the alternative route, as shown in
\Cref{fig:ue-so-usage} (right). In contrast, the Linear Model uses the
alternative route a lot more.  We note that this is to be expected, as
discussed in \Cref{sec:psychmods}.


\section{Conclusion}

Besides providing a framework for formalizing strategic routing
scenarios, we gave different algorithms solving SAP.
Both of these contributions open the door to future research.
Concerning SAP, we have seen that different psychological models can
lead to different alternative routes, and it would be interesting to
study how people actually behave depending on the exact formulation of
the suggestion and on potential additional incentives to take a longer
route.
It is promising to study models that lie in-between the User
Equilibrium and the Linear Model.  By setting, e.g.,
$c(x) = \tanh(a \cdot x / d)$ in the Quotient Model
(\Cref{eq:general-model-def}), we obtain a model that behaves like the
Linear Model for small $x$ and approaches the User Equilibrium Model
for larger $x$, where the constant $a$ controls how quickly that
happens.  We note that this choice of $c(x)$ satisfies the conditions
of~\Cref{thm:quotient-model-pareto-conform}, implying that the
resulting model is Pareto-conform, which makes the algorithms
from~\Cref{sec:sap-algos} applicable.
Concerning algorithmic performance, we have seen that our
proof-of-concept implementation yields reasonable run times.  Our
implementation uses techniques such as A* to speed up computation.
Beyond that, there is still potential for engineering, e.g., by
employing preprocessing techniques.
Beyond the SAP problem, our framework gives rise to various problems
in the context of strategic routing that are worth studying
algorithmically.

\FloatBarrier\
\bibliography{meta/ms}
\end{document}